\documentclass[runningheads]{llncs}
\usepackage{tikz}
\usepackage{listings}
\usepackage{xcolor}
\usepackage{xspace}
\usepackage{multirow}
\usepackage[textsize=tiny]{todonotes}
\usepackage{amsmath}
\usepackage{stmaryrd}
\usepackage{url}
\usepackage{paralist}
\usepackage[hidelinks]{hyperref}
\usepackage[capitalize]{cleveref}
\usepackage{colortbl}
\usepackage{fp}
\usepackage{enumitem}
\usepackage{placeins}
\usepackage{float}
\usepackage{stfloats}
\usepackage[frozencache,cachedir=minted-cache]{minted}
\usepackage{amssymb}
\usepackage{pgfplots}
\usepackage{xfp}
\usepackage[numbers,sort&compress]{natbib}
\usepackage{booktabs}
\usepackage{paralist}
\usepackage{lineno}
\usepackage[hidelinks]{hyperref}            % the badge is a link to the actual DOI
\usepackage[firstpage]{draftwatermark}      % free badge placement
\usepackage{orcidlink}

\definecolor{codehighlight}{RGB}{220, 255, 220}
\setminted{breaklines, frame=none, framesep=3mm,  % linenos, % numbersep=-2.5mm,
           mathescape=true, escapeinside=||, highlightcolor=codehighlight} % linenos, frame=leftline
\definecolor{operator}{RGB}{0, 102, 187}
\newcommand{\op}[1]{\textcolor{operator}{#1}}

\crefname{figure}{Figure}{Figures}
\crefname{table}{Table}{Tables}
\crefname{example}{Example}{Examples}
\crefname{section}{Section}{Sections}
\crefname{definition}{Definition}{Definitions}

\usetikzlibrary{positioning, arrows.meta, matrix, calc}

\definecolor{codegray}{rgb}{0.5,0.5,0.5}
\definecolor{backcolour}{rgb}{0.95,0.95,0.92}

\lstdefinestyle{mystyle}{
    numberstyle=\tiny\color{codegray},
    basicstyle=\ttfamily\footnotesize,
    breakatwhitespace=false,
    keepspaces=true,
    showspaces=false,
    showstringspaces=false,
    numbers=left,
    numbersep=5pt,
    lineskip=0pt,
    aboveskip=0pt,
    belowskip=0pt,
    escapeinside={<@}{@>}
}

\tikzset{   every node/.style={minimum height=5mm, text height=1.5ex, text depth=.25ex},
            level distance=12mm}

\newcommand{\ignore}[1]{}
\newcommand{\Digests}{{\mathcal{A}}}

\newcommand{\D}{\mathcal{D}}

\newcommand{\semT}[1]{\llbracket #1 \rrbracket}

\newcommand{\aSem}[1]{\left\llbracket #1 \right\rrbracket^\sharp}

\newcommand{\aSemDigest}[1]{\left\llbracket #1 \right\rrbracket^\sharp_{\Digests}}

\newcommand{\digestAnnotation}[1]{ {\textcolor{gray}{\ensuremath{#1} } }}

\newcommand{\raceUnknown}[1]{\mathcal{R}_{#1}}

\newcommand{\act}{\textsf{act}}
\newcommand{\create}{\textsf{create}}
\newcommand{\newNoArgs}{\textsf{new}}

\newcommand{\lock}{\textsf{lock}}
\newcommand{\unlock}{\textsf{unlock}}

\newcommand{\abs}[1]{\lvert#1\rvert}

\newcommand{\Globals}{\mathcal{G}}

\newcommand{\Nodes}{\mathcal{N} }
\newcommand{\Edges}{\mathcal{E} }
\newcommand{\Actions}{\mathcal{A}ct}
\newcommand{\Traces}{{\mathcal{T}}}

\newcommand{\TIDs}{\textsf{TID}}

\newcommand{\main}{\texttt{main}}

\newcommand{\Mutexes}{\mathcal{S}}

\newcommand{\unique}{\textsf{unique}}
\newcommand{\mayrun}{\textsf{may\_run}}
\newcommand{\tid}{\textsf{tid}}
\newcommand{\Bool}{\textsf{Bool}}
\newcommand{\maintid}{\textsf{main}}
\newcommand{\othertid}{\textsf{other}}
\newcommand{\false}{\textsf{false}}
\newcommand{\true}{\textsf{true}}

\newcommand{\semDigests}[1]{\aSem{#1}_\Digests}

\newcommand{\newDigests}{\new^\sharp_\Digests}

\newcommand{\initDigests}{\init^\sharp_\Digests}

\newcommand{\new}{\textsf{new}}
\newcommand{\initMutex}{\textsf{init}}

\newcommand{\init}{\textsf{init}}
\newcommand{\MT}{\ensuremath{\textsf{MT}}}

\newcommand{\STmain}{\ensuremath{\textsf{ST}_\textsf{main}}}
\newcommand{\MTmain}{\ensuremath{\textsf{MT}_\textsf{main}}}

\newcommand{\Goblint}{\textsc{Goblint}\xspace}
\newcommand{\SVCOMP}{\textsc{Sv-Comp}\xspace}

\newtheorem{proofsketch}{Proof Sketch}

\newcommand{\ltDrawingDefs}{
    \tikzset{
        every node/.style={node distance=40pt and 30pt},
        programpoint/.style={circle,draw,font=\small},
        programpointsink/.style={programpoint,line width=0.5mm},
        programpointwide/.style={programpoint,node distance=40pt and 55pt},
        programpointnarrow/.style={programpoint,node distance=40pt and 15pt},
        programpointwidesink/.style={programpointwide,line width=0.5mm},
        edgelabel/.style={midway,font=\small,above=0.5mm}
    }
    \newcommand{\programo}[3]{
        \draw[->](##1)--(##3)node[edgelabel]{$##2$};
    }
    \newcommand{\programon}[4][]{
        \node[programpoint##1,right= of ##2](##4){};
        \draw[->](##2)--(##4)node[edgelabel]{$##3$};
    }
    \newcommand{\programond}[4][]{
        \node[programpoint##1,right= of ##2,dotted](##4){};
        \draw[->,dotted](##2)--(##4)node[edgelabel]{$##3$};
    }
    \newcommand{\createo}[3][]{
        \draw[-latex,blue](##2)--(##3)node[edgelabel,##1]{$\to_c$};
    }
    \newcommand{\joino}[3][]{
        \draw[-latex,brown](##2)--(##3)node[edgelabel,##1]{$\to_j$};
    }
    \newcommand{\mutexo}[5][]{
        \draw[-latex,red,##5](##2) to node[edgelabel,##1]{$\to_{##4}$} (##3);
    }
    \newcommand{\signalo}[5][]{
        \draw[-latex,OliveGreen,##5](##2) to node[edgelabel,##1]{$\to_{##4}$} (##3);
    }
}

\begin{document}


\title{Data Race Detection by\\ Digest-Driven Abstract Interpretation}
\titlerunning{Data Race Detection by Digest-Driven Abstract Interpretation}

\newcommand*{\extended}{}%

\ifdefined\extended
\subtitle{(Extended Version)}
\fi

\author{Michael Schwarz\inst{1}\orcidlink{0000-0002-9828-0308} \and Julian Erhard\inst{2,3}\orcidlink{0000-0002-1729-3925}}
\authorrunning{M. Schwarz and J. Erhard}

\institute{
     National University of Singapore, Singapore\\
     \email{m.schwarz@nus.edu.sg} \and
     Technische Universit\"at M\"unchen, Garching, Germany\and
     Ludwig-Maximilians-Universit\"at M\"unchen, Munich, Germany\\
     \email{julian.erhard@tum.de}
}

\maketitle              % typeset the header of the contribution
\SetWatermarkAngle{0}
\SetWatermarkText{\raisebox{14.5cm}{%
  \hspace{10.3cm}%
  \href{https://doi.org/10.5281/zenodo.17128591}{\includegraphics[width=20mm,keepaspectratio]{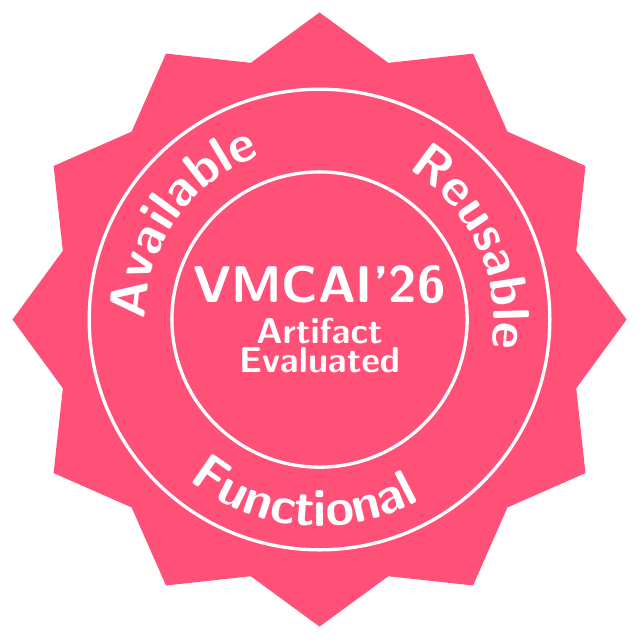}}%
}}

% \begin{abstract}
% Sound static analysis, unlike other techniques, can establish that a program is free of data races.
% %
% % Data races continue to be one of the major sources of severe bugs in concurrent programs. %, and sound static analysis is a key technique to root them out.
% At its core, this amounts to showing that no conflicting accesses to the same memory location
% can happen concurrently.
% We provide a framework for reasoning about which accesses cannot happen simultaneously.
% We repurpose the concept of \emph{digests} that brings tunable concurrency-sensitivity to thread-modular
% value analysis by abstract interpretation.
% There, digests---abstractions of computational pasts---are used to decide which \emph{data-flow} needs to be accounted for.
% %is possible between abstractions of threads.
% %
% We extend this framework and use digests to determine which accesses may be racy.
% To this end, we propose a formal definition of a data race in the thread-modular \emph{local trace} semantics.
% Then, we show how criteria for excluding potential conflicts can be expressed as race digests.
% %
% Lastly, we report on our implementation of digest-driven data race detection in the static analyzer
% \Goblint, and evaluate it on the \SVCOMP suite:
% combining the lockset digest with digests reasoning about thread \emph{id}s and joins increases
% the number of solved tasks more than fivefold.
% \keywords{data races, concurrency, static program analysis, software verification, abstract interpretation}
% \end{abstract}

\begin{abstract}
Sound static analysis can prove the absence of data races by establishing that no two conflicting memory accesses can occur at the same time.
We repurpose the concept of \emph{digests}---summaries of computational histories originally introduced to
bring tunable concurrency-sensitivity to thread-modular value analysis by abstract interpretation,
extending this idea to race detection: We use digests to capture the conditions under which conflicting
accesses may not happen in parallel.
To formalize this, we give a definition of data races in the thread-modular local trace semantics and
show how exclusion criteria for potential conflicts can be expressed as digests.
We report on our implementation of digest-driven data race detection in the static analyzer
\Goblint, and evaluate it on the \SVCOMP benchmark suite.
Combining the lockset digest with digests reasoning on thread \emph{id}s and thread joins increases
the number of correctly solved tasks by more than a factor of five compared to lockset reasoning alone.
\keywords{data races, concurrency, static program analysis, software verification, abstract interpretation}
\end{abstract}

\section{Introduction}
Data races % are concurrency bugs that
happen when the same data is accessed by two different threads,
the accesses are unordered, and at least one access is a write. As data races and any bugs they may trigger
do not manifest deterministically,
finding them by testing is exceedingly difficult. For this reason, many
dynamic~\cite{OCallahanC03,RamanZSVY12,SavageBNSA97,Effinger-DeanLCGB12,MarinoMN09,SerebryanyI09} and
static~\cite{PratikakisFH06,Terauchi08, HeSF22, DietschHKSP23, EnglerA03,LiuLLTS021} analyses have been proposed,
with some of the latter~\cite{VojdaniARSVV16,RacerD,dacik2025racerflightweightstaticdata,ChopraPD19,kastner2017finding,mine2016taking} building
on abstract interpretation~\cite{CousotC77}.

Showing data race freedom amounts to proving at least one of two properties for each pair of accesses
where at least one is a write. Either
\begin{itemize}
    \item[(a)] a separation in space (accesses go to different data); or
    \item[(b)] a separation in time (accesses are not concurrent)
\end{itemize}
is to be established.
Here, we propose a novel framework for property $(b)$ and a principled way of designing new abstractions of it.
\begin{figure}[t]
    \centering
    \input{figures/program-mutexes-running-example.tex}
    \caption{Multi-threaded program with thread prototypes \textsf{main} and $t1$.}
    \label{l:running-example}
\end{figure}
\begin{example}\label{ex:running-example}
    Consider the program in \cref{l:prog1} with thread prototypes \textsf{main} and $t1$. % given in program form.
    It uses the mutex $a$ for synchronization.
    The mutex is first initialized by \textsf{main}.
    Later, \textsf{main} creates a thread from prototype $t1$.
    There are two write-accesses to the global variable $g$ in \textsf{main}, and one in $t1$.
    While the access to $g$ in line \ref{line:firstAccessInMain} is unprotected, i.e., $a$ is not held, it occurs before $t1$ is created.
    Thus, the program is still single-threaded at that time and the access does not race with any other access.
    The accesses in lines \ref{line:secondAccessInMain} and \ref{line:accessInT1} do not race as both are protected by $a$.
    $\hfill\qed$
\end{example}
\cref{ex:running-example} illustrates that one has to rely on different forms of reasoning
to establish race freedom.
% on why certain pairs of accesses are not racy.
Sometimes, the program is still in single-threaded mode for an access, sometimes the set of mutexes overlaps.
In other cases, two writes from the same thread prototype can be shown to not race because only
one unique instance of that thread prototype is created.
Here, we introduce a framework that captures all these different types of arguments.
Conceptually,
\begin{compactitem}
\item[(a)] we associate with each access an abstraction of the execution history leading up to the access
from some set $\Digests$ of \emph{digests}. Locksets, thread \emph{id}s, and whether the program is single-threaded all are (rather coarse) history abstractions. % of the execution history.
\item[(b)] we introduce, for each global $g$, a predicate $(||^{?}_g): \Digests \to \Digests \to \{\textsf{false}, \top \}$ which checks whether
two accesses with respective digests may happen in parallel. While $\textsf{false}$ indicates that the accesses may definitely not happen in parallel,
$\top$ indicates that either such accesses may happen in parallel, or that the analyzer cannot exclude that they may based on digests alone.
Whenever the definition does not depend on $g$, we write $||^{?}$, omitting the index.
\end{compactitem}
\begin{example}\label{ex:lockset-intro}
    For a digest tracking locksets, one can, irrespective of the global, set
    \[
        A_a\;||^{?}\;A_b = \begin{cases}
            \textsf{false} & \text{ if } A_a \cap A_b \neq \emptyset\\
            \top & \text{otherwise}
        \end{cases}
    \]
    to capture that accesses with overlapping locksets do not happen in parallel.
    $\hfill\qed$
\end{example}
As a semantic foundation for our framework, we build on the \emph{local trace semantics}~\cite{schwarz2023clustered,ErhardBHKSSSTV2025,DBLP:conf/sas/SchwarzSSAEV21,hammerandnail}
which is a concrete semantics which, for each thread, tracks only \emph{relevant} history information.
While commonly some sort of interleaving semantics~\cite{Lamport1979} is used to
characterize multi-threaded programs, the local trace semantics as an equivalent, but thread-modular concrete semantics,
is better-suited to our approach, as digests also take a thread-modular view of the execution history.

Previous work~\cite{schwarz2023clustered} introduced the notion of \emph{digests} as a way to track relevant history information in a
thread-modular way to make value analyses by abstract interpretation more precise.
In this way, one obtains \emph{families} of abstract interpretations parametrized
by digests which offer a configurable degree of what was later dubbed \emph{concurrency-sensitivity}~\cite{hammerandnail}.
However, neither of these works addresses data races.
We take on this gap and repurpose digests for data race detection.
The existing digest-driven abstract interpretation framework provides a notion of \emph{admissibility} for instances of
digests from which the soundness of analyses refined by such an instance directly follows.
Here, we enhance this framework with a predicate $||^?_g$ as discussed above and
show how a first definition of such a predicate can automatically be derived from an admissible digest.

Our proposed framework does not only drastically simplify reasoning about race freedom
by decoupling the base analysis from the admissibility of the digests, but also
provides the conceptual foundation for our existing implementation of data race detection
in the successful static analyzer \Goblint.

The rest of the paper is structured as follows:
\cref{s:lt} recalls the local trace semantics and its abstract interpretation, whereas \cref{s:digests}
recalls digests and introduces the notion of \emph{access stability}.
\cref{s:lt-traces} proposes two definitions of data races,
% for local traces
one closely aligned with intuition
and one % in terms of bidirectional trace compatibility
more directly amenable to abstraction,
and establishes their equivalence.
\cref{s:races-digest} formally introduces the predicate $||^?_g$ and its soundness requirements, before
showing how a definition can automatically be derived and how bespoke definitions can sometimes be more precise.
\cref{s:algorithm} presents our digest-driven race detection algorithm,
% \cref{s:special-digests} introduces digests tailor-made to support additional features of \pthreads,
whereas \cref{s:experiments} reports on experiments within the static analyzer \Goblint.
\cref{s:related} outlines related work, with \cref{s:conclusion} describing future work and concluding.

\section{Local Trace Semantics and Its Abstract Interpretation}\label{s:lt}
Thread-modular abstract interpretation scales as it avoids considering all possible interleavings.
Traditional concrete semantics, on the other hand, distinguish all interleavings.
This disconnect renders soundness proofs of thread-modular analyses cumbersome.
The local trace semantics~\cite{schwarz2023clustered,ErhardBHKSSSTV2025,DBLP:conf/sas/SchwarzSSAEV21,hammerandnail}, which has been shown equivalent
to the interleaving semantics w.r.t. safety properties~\cite{ErhardBHKSSSTV2025}, embraces thread-modularity already at the level of the concrete semantics:
\begin{compactitem}
    \item
    Instead of considering interleavings and hence total orders, the local trace semantics considers a \emph{partial order} of actions.
    Other than actions within one thread, only relevant actions  are ordered, namely those where information flows between threads
    as is, e.g., the case for locking and unlocking mutexes.
    \item A local trace corresponds to the \emph{local perspective} of one thread.
    It only contains those parts of the computational history that the thread can have knowledge about:
    This includes its own past as well as those parts of the past of other threads that it has learned about by communicating.
    Such communication happens via pairs of \emph{observable} and \emph{observing actions}, such as unlocking and locking some specific mutex.
\end{compactitem}
At an observing action (e.g., locking a mutex), the observed local trace (e.g., one ending in an unlock of the same mutex)
is incorporated into the local trace of the observing thread and the result prolonged to record the observing action
having taken place.
At this step, only \emph{compatible} local traces can be incorporated.
Abstract interpretations of the local trace semantics track values for program points and for observable actions
and thus share the rough structure with the concrete semantics~\cite{DBLP:conf/sas/SchwarzSSAEV21,schwarz2023clustered}
allowing for scalable analyses that enjoy modular soundness proofs.
Abstractions of the aforementioned compatibility of local traces are a key ingredient for making thread-modular abstract interpretation and, as we show in this paper, data race detection more precise.

\paragraph{Local Traces.}
More concretely, a local trace consists of several swimlanes, one for each thread appearing in
the trace, which records a sequence of thread-local configurations attained by this thread.
Such thread-local configurations record, e.g., the values of \emph{local} variables and the current program point of the thread,
but crucially no information about the global state, such as values of global variables.
The order of configurations within one swimlane is the program order of the corresponding thread.
Additionally, some dependencies between thread-local configurations are recorded.
Such dependencies
may, e.g., be between the action creating a thread\footnote{
    For technical reasons relating to the use of creation histories as concrete thread \emph{id}s in prior work, the dependency is from the last configuration of the creating thread before the
    create action to the first configuration of the created thread.
} and the first configuration of the resulting thread,
or between unlock and lock operations of the same mutex.
We call the union of all such dependencies and the program orders the causality order, and demand that it is
acyclic and has a unique maximal element. This unique maximal configuration belongs to one thread
---the \emph{ego thread}--- whose perspective is taken.
\setminted{
	escapeinside=@@
}
\newcounter{lineno}
\newcommand{\lnum}{\refstepcounter{lineno}\tiny\arabic{lineno}\;\;}
\newcommand{\lnumref}[1]{\ref{#1}}
\begin{figure}[t]
	\begin{minipage}[t]{0.5\textwidth}
		\begin{tabular}[t]{rl|c|c} % r = line number column
			& & MT-dig. & Lockset-dig. \\
			\lnum & \mintinline{c}{main:} & \digestAnnotation{\STmain} & \digestAnnotation{\emptyset} \\
			\lnum & \mintinline{c}{    @\op{init}@(a);} & \digestAnnotation{\STmain} & \digestAnnotation{\emptyset} \\
			\lnum \label{line:mg:firstAccessSequenceInMain} & \mintinline{c}{    @\op{init}@(@$m_g$@);} & \digestAnnotation{\STmain} & \digestAnnotation{\emptyset} \\
			\lnum & \mintinline{c}{    @\op{lock}@(@$m_g$@);}  & \digestAnnotation{\STmain} & \digestAnnotation{\{m_g\}} \\
			\lnum \label{line:mg:firstAccessInMain} & \mintinline{c}{    g = 0;}  & \digestAnnotation{\STmain} & \digestAnnotation{\{m_g\}} \\
			\lnum & \mintinline{c}{    @\op{unlock}@(@$m_g$@);} & \digestAnnotation{\STmain} & \digestAnnotation{\emptyset} \\
			\lnum & \mintinline{c}{    @\op{create}@(t1);} & \digestAnnotation{\MTmain} & \digestAnnotation{\emptyset} \\
			\lnum \label{line:mg:secondAccessSequenceInMain} & \mintinline{c}{    @\op{lock}@(a);} & \digestAnnotation{\MTmain} & \digestAnnotation{\{a\}} \\
			\lnum & \mintinline{c}{    @\op{lock}@(@$m_g$@);} & \digestAnnotation{\MTmain} & \digestAnnotation{\{a, m_g\}} \\
			\lnum \label{line:mg:secondAccessInMain} & \mintinline{c}{    g = 5;}  & \digestAnnotation{\MTmain} & \digestAnnotation{\{a, m_g\}} \\
			\lnum & \mintinline{c}{    @\op{unlock}@(@$m_g$@);} & \digestAnnotation{\MTmain} & \digestAnnotation{\{a\}} \\
			\lnum & \mintinline{c}{    @\op{unlock}@(a);} & \digestAnnotation{\MTmain} & \digestAnnotation{\emptyset} \\
		\end{tabular}
	\end{minipage}%
	\hspace{1em}%
	\begin{minipage}[t]{0.5\textwidth}
		\begin{tabular}[t]{rl|c|c}
			& & MT-dig. & Lockset-dig. \\
			\lnum & \mintinline{c}{t1:} & \digestAnnotation{\MT} & \digestAnnotation{\emptyset} \\
			\lnum \label{line:mg:accessSequenceInT1} & \mintinline{c}{    @\op{lock}@(a);} & \digestAnnotation{\MT} & \digestAnnotation{\{a\}} \\
			\lnum & \mintinline{c}{    @\op{lock}@(@$m_g$@);} & \digestAnnotation{\MT} & \digestAnnotation{\{a, m_g\}} \\
			\lnum \label{line:mg:accessInT1} & \mintinline{c}{    g = 12;}  & \digestAnnotation{\MT} & \digestAnnotation{\{a, m_g\}} \\
			\lnum & \mintinline{c}{    @\op{unlock}@(@$m_g$@);} & \digestAnnotation{\MT} & \digestAnnotation{\{a\}} \\
			\lnum & \mintinline{c}{    @\op{unlock}@(a);} & \digestAnnotation{\MT} & \digestAnnotation{\emptyset} \\
		\end{tabular}
	\end{minipage}
	\caption{Multi-threaded program from \cref{l:running-example} enhanced with the atomicity mutex $m_g$ required by the local trace semantics.
	All program locations are annotated using the MT-digest (see \cref{f:threadingmainA}) and the lockset-digest (see \cref{f:locksets}).\label{l:prog1}}
\end{figure}
\setminted{
	escapeinside=||
}

\begin{figure*}
    \ltDrawingDefs
    \centering
    \scalebox{0.67}{
        \begin{tikzpicture}
            \node[programpoint](tmpp0){};
            \programon{tmpp0}{\init(a)}{tmpp1}
            \programon{tmpp1}{\init(m_g)}{tmpp11}
            \programon{tmpp11}{\lock(m_g)}{tmpp12}
            \programon{tmpp12}{g = 0}{tmpp13}
            \programon[wide]{tmpp13}{\unlock(m_g)}{tmpp14}
            \programon[wide]{tmpp14}{\create(t1)}{tmpp2}
            \programon{tmpp2}{\lock(a)}{tmpp3}
            \programon{tmpp3}{\lock(m_g)}{tmpp31}
            \programon{tmpp31}{g = 5}{tmpp4}
            \programon[wide]{tmpp4}{\unlock(m_g)}{tmpp5}
            \programon{tmpp5}{\unlock(a)}{tmpp6}

            \node[programpoint,below right=40pt and 42pt of tmpp14](t1pp0){};
            \programon[sink]{t1pp0}{\lock(a)}{t1pp1};

            \mutexo[below=2mm, pos=0.2]{tmpp11}{tmpp12}{m_g}{out=300,in=210} % 80, 110
            \mutexo[below=0mm, pos=0.6]{tmpp14}{tmpp31}{m_g}{out=320,in=190} % 80, 110
            \mutexo[left=5mm]{tmpp1}{tmpp3}{a}{out=290,in=190}
            \mutexo[below=1mm, pos=0.2]{tmpp6}{t1pp1}{a}{out=210,in=65}
            \createo[left=2mm, pos=0.8]{tmpp14}{t1pp0}

        \end{tikzpicture}}
    \caption{Stylized local trace of the program in \cref{l:prog1}, where the ego thread is the single instance of $t1$, which completes its first action \lock$(a)$.
    The black arrows indicate the order of configurations within a thread.
    The red arrows indicate an ordering between two configurations that is mediated via a \unlock/\textsf{init} and \lock\ of a certain mutex.
    The blue arrow indicates the order between a configuration of a thread creating another thread and the first configuration of the newly created thread.
    }\label{f:prog1}
\end{figure*}

\begin{example}
    Consider the program in \cref{l:prog1}, for now ignoring any annotations.
    \cref{f:prog1} shows an example local trace for this program, with the last appearing configuration of the ego thread
    highlighted by a thicker border.
    $\hfill\qed$
\end{example}
The program in \cref{l:prog1} is a variant of the one from \cref{l:running-example} where all accesses to the global
variable $g$ are surrounded by an additional mutex $m_g$, its so-called \emph{atomicity mutex}.
In fact, this is how the local trace semantics enforces a sequential-consistency-like property for accesses to global variables.
Let $\Globals$ be the set of global variables.
Then, any access $\chi$ to a global $g\in\Globals$ (we allow accesses of the form \texttt{x = g} and \texttt{g = x}
to copy the value of a global $g$ from/to a local variable $x$),
is immediately preceded by locking the atomicity mutex for $g$ and immediately followed by unlocking it, i.e.,
it occurs as the sequence $\lock(m_g); \chi; \unlock(m_g)$.
This allows conveniently treating accesses to global variables as local actions---with all communication happening via locking and unlocking atomicity mutexes.
We will outline in \cref{s:lt-traces} how this assumption interacts with the goal of detecting data races.

More formally, let $\Nodes$ denote the set of all program points.
We consider disjoint, uniquely labeled, control-flow graphs where each control-flow graph corresponds to the prototype of a thread.
Control-flow graphs consist of nodes corresponding to program points and edges
$(u,\act,v) \in \Edges$ where $u,v \in \Nodes$ and $\act$ is from some set $\Actions$ of actions.
This set of actions is given as the disjoint union of the sets of observing, observable, local, and creating actions.
Each control-flow graph has exactly one start node which has no incoming
edges. We identify the label of a control-flow graph with its start node.
One of the control-flow graphs is labeled \textsf{main}. It represents the main thread with which execution starts.

Let us call the set of all local traces $\Traces$.
Within $\Traces$, there is a subset $\init \subset \Traces$ of local traces consisting only of an initial configuration of the main thread
at its initial program point $u_0$.

An instantiation of the local trace semantics to some set of actions then provides for each edge $(u,\act,v)$ a function
$\semT{(u,\act,v)}$ prolonging a local trace
with the corresponding action, where possible.
For non-observing actions, $\semT{(u,\act,v)}$ has type $\Traces \to 2^{\Traces}$ where the returned set is either a
singleton or empty.\footnote{
    This is purely for technical reasons. Constructs such as non-deterministic assignment remain possible,
    as CFGs can have several outgoing edges per node.
}
This corresponds to extending the local trace with the action $\act$ when possible.
For example, extending a local trace $t$ along an edge originating at some program point $u$ is only possible if $t$
ends in $u$.
For an observing action $\act$, $\semT{(u,\act,v)}$ has type $\Traces \to \Traces \to 2^{\Traces}$,
where the returned set is again either empty or a singleton.
When $\semT{(u,\act,v)}\,t_0\,t_{1}$ returns a non-empty set, the local traces $t_0$ and $t_1$ are said to be \emph{compatible} w.r.t. the edge $(u, \act, v)$ and each other.
% The local traces that $\semT{(u,\act,v)}$ receives as arguments are said to be \emph{compatible w.r.t. the edge $(u, \act, v)$} in case the set that it returns is non-empty.
In particular, to be compatible, $t_1$ must end in an observable action whose type matches the ones observed by $\act$.
Depending on the type of observing action, additional requirements have to be fulfilled.
For instance, the operations corresponding to locking and unlocking some mutex $a$ need to be totally ordered.
For edges corresponding to the creation of new threads, which are of the form $(u_1,\create\,u',u_2)$,
we additionally require a function
$\new: \Traces \to \Nodes \to \Nodes \to 2^\Traces$;
$\new\,t\,u_1\,u'$ computes from a local trace $t$ ending in $u_1$ the local trace of the newly created thread starting its execution at $u'$, provided such a thread creation is possible at this edge.

\paragraph{Instantiation.}
While we avoid fixing the set of observable, observing, and local actions to remain general,
we will fix subsets to give examples.
We then have, in addition to thread creates as described above, the following actions:
\begin{compactitem}
    \item The observing action $\lock(a)$, and the observable actions $\unlock(a)$ and $\initMutex(a)$,
    for each mutex $a$ from the set of mutexes $\Mutexes$. We require $\{ m_g \mid g \in \Globals \} \subseteq \Mutexes$,
    i.e, atomicity mutexes for all global variables are included in the set of mutexes.
    \item As local actions, we have at least $x=g$ and $g=x$ to copy values between local variables and globals.
    We call either an \emph{access} to $g$.
\end{compactitem}
Consider again the local trace in \cref{f:prog1} which originates from the program in \cref{l:prog1}.
This local trace belongs to thread $t1$ as evidenced by the maximal element being a configuration of that thread.
Of particular interest are the locking orders $\to_a$ and $\to_{m_g}$ for mutexes $a$ and $m_g$,
which illustrate additional conditions that need to hold in local traces:
In the lock order $\to_a$ for each mutex $a$, every lock must be preceded by
exactly one unlock of $a$ or be the first lock of $a$ in which case it is preceded by \textsf{init(a)}.
Each unlock and init operation of $a$ must be directly followed by at most one lock operation.
Some further requirements on these orders are detailed in previous work \cite{schwarzthesis}.

\paragraph{Fixpoint Formulation.}
While the set $\Traces$ of local traces can be characterized as the least fixpoint obtained when, starting from $\init$,
applying the functions $\semT{\cdot}$ and $\new$ to all (pairs of) local traces and taking the union
with $\init$, a different---but equivalent~\cite{schwarzthesis,hammerandnail}---view is more conducive to static analysis:
Here, one considers the least solution of a constraint system with several
constraint system variables, called \emph{unknowns}.
One such unknown is introduced for each program point, and for each observable action.
An unknown $[u]$ associated with a program point $u$ collects all local traces ending in $u$.
For each observable action $\act$ an unknown $[\act]$
collects all local traces ending with this action.
When lifting $\semT{\cdot}$ and \new\ point-wise
to sets, the resulting constraint system takes the following form:
\begin{align}
[u_0] &\supseteq \init \label{eq:cons-init}\\
[u'] &\supseteq \new([u_1],u_1,u'), \text{for } (u_1,\create\,u',u_2) \in\Edges \label{eq:cons-create}\\
[u_2] &\supseteq \semT{(u_1,\act,u_2)}([u_1]), \text{for } (u_1,\act,u_2) \in\Edges, \act \text{ not observing} \label{eq:cons-control-flow} \\
[\act] &\supseteq \semT{(u_1,\act,u_2)}([u_1]), \text{for } (u_1,\act,u_2) \in\Edges, \act \text{ observable} \label{eq:cons-observable-action} \displaybreak[0]\\
[u_2] &\supseteq \semT{(u_1,\act,u_2)}([u_1],[\act']),  \text{for } (u_1,\act,u_2) \in\Edges, \act \text{ observing}, \; \label{eq:cons-observing-action}\\
&  \qquad\qquad\qquad\qquad \qquad\qquad\qquad\qquad \act' \text{ observed by } \act \nonumber
\end{align}
For edges labelled with an observable action, there are two constraints:
Constraints (\ref{eq:cons-control-flow}) propagate
sets of resulting local traces to the unknown corresponding to the control-flow successor, whereas
constraints (\ref{eq:cons-observable-action})
ensure that the resulting local traces are recorded at the unknown corresponding to the observable action.
% In the concrete semantics, the right-hand sides for these two constraints are identical.

\ignore{The least solution of this constraint system thus describes the set of local traces in a way immediately amenable to an
abstract interpretation which is thread-modular, while remaining flow-sensitive within each thread.}

\paragraph{Abstract Interpretation.}
To set up an abstract interpretation, it suffices to let unknowns range over values from a suitable abstract domain $\D$
instead of over $2^\Traces$
and replace functions $\semT{\cdot}$, $\init$, and $\new$  with suitable abstract counterparts
$\aSem{\cdot}$, $\init^\sharp$, and $\new^\sharp$.
Adapting constraints \labelcref{eq:cons-init,eq:cons-control-flow,eq:cons-observable-action,eq:cons-observing-action,eq:cons-create} in this way,
one obtains the abstract constraint system.
For simplicity, we assume that $\D$ is a complete lattice
with least element $\bot$ (denoting unreachability),
greatest element $\top$, and that $\sqcup$ is used to denote the join (binary least-upper-bound).
Furthermore, we require a monotonic concretization function $\gamma_{\D}: \D \to 2^\Traces$.
An analysis is \emph{sound} when $\aSem{\cdot}$, $\init^\sharp$, and $\new^\sharp$ are sound abstractions (w.r.t.\ $\gamma_{\D}$)
of their concrete counterparts.
Our considerations in this paper are generic in the base analysis and abstract domain.
We assume that it is sound but do not provide details.

\section{Digests as Concurrency-Sensitivity}\label{s:digests}
Digests~\cite{schwarz2023clustered,hammerandnail,schwarzthesis} are a way to refine abstract interpretations of
multi-threaded programs to yield more precise
results.
Unknowns are split according to an abstraction of the computational past (the \emph{digest}).
Then, the abstract value at each unknown only describes those local traces that agree w.r.t. the digest.
This approach can be understood as a generalization of trace partitioning~\cite{Handjieva98,Mauborgne05,Rival07,Montagu21}
to a multi-threaded setting.
It offers a convenient way to specify which aspects of the computational history to consider and what to abstract away.

Digests exploit that, when at observing actions different views about the computational past are combined,
some views are known to be incompatible already after considering their abstraction as given in the digest alone.
These combinations then do not need to be considered by the analysis.
Technically, digests overapproximate the concrete trace compatibility:
If two digests are \emph{incompatible}, so are all traces they represent.
However, the reverse does not necessarily hold: digests may be compatible when the corresponding traces are, in fact, not.

More formally, let $\Digests$ be a set and
$\alpha_\Digests:\Traces {\to} \Digests$ a function to extract such information from a local trace.
Given a local trace $t$, we refer to $\alpha_\Digests\,t$ as its \emph{digest}.
We require that for each observable action $\act\in\Actions$, there is a function $\semDigests{\act}: \Digests\to\Digests\to 2^\Digests$
overapproximating the set of resulting digests of traces obtained by executing $\act$.
Similarly, for non-observable actions $\act$ and functions $\semDigests{\act}: \Digests\to 2^\Digests$.
More formally, for an edge $e = (u,\act,v)$ we require
\begin{equation}
    \begin{array}{rlr}
        \forall t_0,t_{1} \in \Traces:
        &\alpha_\Digests(\semT{e}(t_0,t_{1}))
        \subseteq \semDigests{\act}(\alpha_\Digests\,t_0,\alpha_\Digests\,t_{1}) & (\act \text{ observable})\\
        \forall t_0 \in \Traces:&\alpha_\Digests(\semT{e}(t_0))
        \subseteq \semDigests{\act}(\alpha_\Digests\,t_0) & (\act \text{ non-observable})
    \end{array}
    \label{def:ASound}
\end{equation}
where $\alpha_\Digests$ is lifted element-wise to sets.
Intuitively, this corresponds to $\aSemDigest{\cdot}$
soundly overapproximating $\semT{\cdot}$.
Additionally, we require $\semDigests{\act}$ to be deterministic, i.e.,
to either yield an empty set of digests or a singleton set:\footnote{
    This assumption simplifies reasoning about digests. To lift this restriction,
    one can instead consider \emph{abstract digests}~\cite{hammerandnail} where digests
    come from (complete) lattices. As (normal) digests are already quite expressive, we stick
    with them here for simplicity.
}
\begin{equation}
    \begin{array}{rlr}
    \forall A_0,A_1 \in \Digests:&
    \abs{\semDigests{\act}(A_0,A_{1})}\leq 1 & (\act \text{ observable})\\
    \forall A_0 \in \Digests:&
    \abs{\semDigests{\act}(A_0)}\leq 1 & (\act \text{ non-observable})
    \label{def:ADet}
    \end{array}
\end{equation}
As an abstract counterpart for $\newNoArgs$, we require a function $\newDigests: \Digests\to\Nodes\to 2^\Digests$
that returns for a thread starting execution at program point $u_1$ the digest of this new trace.
We once again require determinism in the sense discussed above.
\begin{equation}
    \begin{array}{l}
        \forall t_0 {\in} \Traces:
        \alpha_\Digests(\new\,t_0\,u\,u_1) \subseteq \newDigests\,{(\alpha_\Digests\,t_0)}\,u_1
        \quad\;
        \forall A_0 {\in} \Digests: \abs{\newDigests\,A_0\,u_1}\leq 1
    \end{array}
    \label{def:AnewSound2}
\end{equation}
Furthermore, we require that whenever there is a successor digest for the edge corresponding to creating a new thread,
there is also an appropriate digest for the newly created thread, i.e.,
\begin{equation}
    \begin{array}{lll}
    \semDigests{\create(u_1)}(A_0) \neq \emptyset \implies
    \newDigests\,A_0\,u_1 \neq \emptyset
    \end{array}
    \label{def:AnewSound3}
\end{equation}
\noindent For the initial digest at program start, we define:
% This is a set because one could have stupid things such as the values of locals in there
\begin{equation}
    \begin{array}{lll}
        \initDigests &=& \{\alpha_\Digests\,t \mid t \in\init\}
    \end{array}
	\label{def:AinitSound}
\end{equation}
Deviating from earlier work, we here additionally require a property called
\emph{access stability}, which intuitively states that executing an access
sequence of the form $\lock(m_g); x = g; \unlock(m_g)$ or
$\lock(m_g); g = x; \unlock(m_g)$ does not affect the digest, whenever it is defined.
More formally, consider a sequence $\bar \chi$ consisting of an access $\chi$ (read or write) to a global $g$ and
locking and unlocking the atomicity mutex $m_g$.
For such \emph{access sequences} of the form
%edges from the program
% \[
    $(v^i_0,\lock(m_g),v^i_{1}) \cdot (v^i_{1},\chi,v^i_{2}) \cdot (v^i_{2},\unlock(m_g),v^i_{3})$,
% \]
%
we define the effect on the digest by
\[
  \semDigests{\bar \chi}(A_0,A_1) =
  \left(\semDigests{\unlock(m_g)} \circ \semDigests{\chi} \circ \semDigests{\lock(m_g)}\right) (A_0, A_1).
\]
Here, composition is lifted to return $\emptyset$ when an intermediate computation does so, and
we identify singleton sets of resulting digests with their members.
We call a digest \emph{access stable}, if for all access sequences $\bar \chi$ we have
\begin{equation}
  \forall A_0 \in\Digests: \textstyle{\bigcup_{A_1 \in\Digests}} \semDigests{\bar \chi}(A_0,A_1) \subseteq \{ A_0 \}.
  \label{def:Astable}
\end{equation}
\begin{definition}\label{def:digest}
    A set $\Digests$ together with functions $\semDigests{\act}$, $\newDigests$,
    and $\initDigests$ is called \emph{digest}.
    If properties \eqref{def:ASound} - \eqref{def:AinitSound} hold for a digest, it is called \emph{admissible}.
    If property \eqref{def:Astable} also holds, the digest is also called \emph{access stable}.
\end{definition}
We remark that the question of admissibility is independent of the used analysis and the abstract domain, and is
tied to the concrete semantics only.
We take the freedom to refer to both the structure such as in \cref{def:digest} as well as to an element of its set $\Digests$
as a \emph{digest}, as the meaning is clear from the context.

A digest then gives rise to a refined abstract constraint system where each of the unknowns is split into several unknowns,
corresponding to the different digests associated with the reaching traces:
Each unknown $[u]$ for a program point $u$ is replaced with unknowns $[u,A]$ for $A\in\Digests$, and each
unknown $[\act]$ for an observable action $\act$ is replaced with unknowns $[\act,A]$ for $A\in\Digests$.

The refined constraint system then takes the following form:
\[
    \begin{array}{lllr}
        \,[u_0,A_0] &\sqsupseteq& \init^\sharp \qquad&
        (\text{for } A_0 \in \initDigests)\\[0.5em]

        \,\relax[u_2,A'] &\sqsupseteq& \aSem{(u_1,\act,u_2)}(\,\relax[u_1,A_0]) &
        \quad (\text{for } (u_1,\act,u_2) \in\Edges, \act \text{ not observing},\\
        \span\span\span\quad A' \in \semDigests{\act}\,A_0)\\[0.5em]

        \,\relax[\act,A'] &\sqsupseteq& \aSem{(u_1,\act,u_2)}(\,\relax[u_1,A_0]) &
         (\text{for } (u_1,\act,u_2) \in\Edges, \act \text{ observable},\\
         \span\span\span\quad A' \in \semDigests{\act}\,A_0)\\[0.5em]

        \,\relax[u_2,A'] &\sqsupseteq& \aSem{(u_1,\act,u_2)}
        (\,\relax[u_1,A_0], \,\relax[\act',A_1]) \span (\text{for } (u_1,\act,u_2) \in\Edges,\\
        \span\span\span \act \text{ observing},
            \act' \text{ observed by } \act, A' \in \semDigests{\act}\,(A_0,A_{1}))\\[0.5em]

        \,\relax[u',A'] &\sqsupseteq& \new^\sharp(\,\relax[u_1,A_0],u_1,u') &
          (\text{for } (u_1,\create\,u',u_2) \in\Edges,  A' \in \newDigests\,A_0\,u').
    \end{array}
\]
We quickly state the main soundness theorem for such refined analyses:
\begin{proposition}
    Refining a sound analysis with an admissible digest yields a sound refined analysis.
\end{proposition}
\begin{proof}[Sketch]
    The proof proceeds by first proving a refined version of the concrete semantics sound w.r.t.\ the original semantics
    and then showing that the refined analysis is sound w.r.t.\ the refined semantics, provided the original analysis
    is sound w.r.t.\ the original semantics (see~\cite[Chapter 2.3]{schwarzthesis}
    and~\cite{hammerandnail}).
    $\hfill\qed$
\end{proof}
Consider the local trace semantics with the actions described in the previous section.
We provide a first example digest for this setting.
\begin{example}
A lightweight thread \emph{id} digest may track for each thread whether
\begin{compactitem}
\item it is the main thread and has not started any other threads ($\STmain$), or
\item it is the main thread and has started other threads ($\MTmain$), or
\item it is some other thread ($\MT$).
\end{compactitem}
The corresponding definitions are given in \cref{f:threadingmainA}.
$\aSemDigest{\lock(a)}$ exploits
that if the digest for the ego thread is $\STmain$, i.e., the ego thread is the main thread and has not created other threads yet, it cannot acquire a mutex from another thread (with the $\MT$ digest) or from itself already having created threads (with the $\MTmain$ digest).
Thus, in these cases, it yields $\emptyset$.
However, the result of $\aSemDigest{\lock(a)}(\MTmain,\STmain)$ cannot be set to $\emptyset$, as the last unlock of the mutex $a$ can
in fact have happened while the main thread was single-threaded, even if further threads were created in the meantime.
This digest is \emph{access stable}.
$\hfill\qed$
\end{example}
\begin{figure*}[t!]
    \begin{minipage}[t]{.4\linewidth}\[
        \begin{array}{lll}
        \Digests = \{ \STmain, \MTmain, \MT \}\\[0.5ex]
		\initDigests = \{\STmain\}\\[0.5ex]
        \newDigests\,M\,u_1 = \{\MT\}
        \end{array}%
    \]
    \end{minipage}%
    \begin{minipage}[t]{.4\linewidth}\[
        \begin{array}{lll}
        \semDigests{\create\,u_0}\,M =
        \begin{cases}
            \{\MT\} & \text{if } M = \MT\\
            \{\MTmain\} & \text{otherwise}\\
        \end{cases}
        \\[0.5ex]
		\semDigests{\act}\,M = \{M\} \qquad\qquad \text{(other non-observing)}\qquad\;\;
        \end{array}%
    \]
    \end{minipage}\\~\\
    \begin{minipage}[t]{\linewidth}\[
        \begin{array}{lll}
            \semDigests{\lock(a)}\,(M_0,M_1) &=& \begin{cases}
                % \{\STmain\} & \text{if } M_0 = \STmain \land M_1 = \STmain\\
                \emptyset & \text{if } M_0 = \STmain \land M_1 \neq \STmain\\
                \{M_0\} & \text{otherwise}
            \end{cases} % \\[3ex]
            % \semDigests{\join(x')}\,(M_0,M_1) &=& \begin{cases}
            %     \emptyset & \text{if } M_0 = \STmain \lor M_1 = \STmain\\
            %     \emptyset & \text{if } M_0 = \MTmain \land M_1 = \MTmain\\
            %     \{M_0\} & \text{otherwise}
            % \end{cases}
        \end{array}
        \]
    \end{minipage}
    \caption{Digest for lightweight thread \emph{id}s.
    }\label{f:threadingmainA}
\end{figure*}

\begin{example}
    Consider the program in \cref{l:prog1}. The MT-digests reachable after each statement is executed are annotated.
    Here, for each program point only one digest is reachable.
    Reasoning about data races using digests, we will later be able to conclude that line \ref{line:mg:firstAccessInMain}
    does not race with either of the other accesses.
    $\hfill\qed$
\end{example}
\begin{remark}
By setting $\mathcal{D} = \{\bot, \bullet\}$ with  $\gamma_{\D}(\bullet) = \Traces$ %$\gamma_{\D}(\bot) = \emptyset$ and
one can obtain a trivial, sound analysis. This way, one can analyze a program
using only digests for reasoning. %, provided the program has no pointers.
\end{remark}

\section{Races in Terms of Local Traces}\label{s:lt-traces}
To arrive at techniques for checking data races that are based on a thread-modular abstract interpretation of the
local trace semantics, we first precisely define what exactly constitutes
a race in the context of the local trace semantics.

We first make the following observation about the atomicity mutexes we have introduced:
In race-free programs, any two conflicting accesses % to the same data
are ordered w.r.t.\ each other.
Thus, the extra edges and dependencies introduced in the local trace semantics for atomicity mutexes
only serve to make transitive dependencies between conflicting accesses in race-free executions explicit.
Additionally, an ordering between read operations is introduced. However, as both orderings of these reads are considered,
and executions only differing in the order in which such reads are performed are indistinguishable, all executions
are still captured.
In fact, by virtue of having atomicity mutexes, the local trace semantics assigns a meaning even to
executions which would not traditionally have one, as they contain races.

To detect whether two accesses are racy, it thus suffices to check whether these appear in any local trace
such that they are unordered but for their special atomicity mutex, leading to the following definition:
\begin{definition}[Racy Accesses]
    Let $g$ be a global variable and $\chi_a$ and $\chi_b$ be two accesses to $g$.
    Accesses $\chi_a$ and $\chi_b$ are \emph{racy} if at least one is a write and
    there exists a local trace in which these accesses are unordered
    w.r.t.\ each other when not considering the order provided by the $m_g$ mutex.
\end{definition}
% Crucially NOT all, but only m_g
%
A local trace containing a racy access to $g$ is visualized in \cref{f:prog0lt}.

While this definition corresponds closely to the intuition,
it does not directly lend itself to abstraction.
We thus provide an alternative definition that can be given in terms of the
functions $\semT{\cdot}$ of the local trace semantics, which is better amenable to abstraction,
and prove them equivalent.

The latter definition builds on the following observation: When accesses are unordered when ignoring the $m_g$ mutexes,
one can, from a local trace containing both accesses and ending in one of them,
construct a local trace where the order of these accesses is flipped.

We first introduce some helpful notation: Consider once again an access sequence
% consisting of an access to a global $g$ and
% locking and unlocking the atomicity mutex $m_g$.
% Such a sequence has the form:
% More formally, consider two sequences $\bar a, \bar b$, each corresponding to an access to a global $g$ and locking and unlocking of the atomicity mutexes $m_g$.
% Each sequences $\bar \chi$ then has the form:
%edges from the program
$\bar \chi \equiv (v_0,\lock(m_g),v_{1}) \cdot (v_{1},\chi,v_{2}) \cdot (v_{2},\unlock(m_g),v_{3})$
% Here $a_i$ are accesses to the global $g$ (either read or write).
where $\chi$ is an access to the global $g$ (either read or write).
We, akin to the definition of $\aSemDigest{\bar\chi}$, denote
% $\bar \chi = (v_0,\lock(m_g),v_{1}) \cdot (v_{1},\chi,v_{2}) \cdot (v_{2},\unlock(m_g),v_{3})$
by $\semT{\bar \chi}(t_0,t_1)$ the function
$(\semT{(v_{2},\unlock(m_g),v_{3})} \circ \semT{(v_{1},\chi,v_{2})} \circ \semT{(v_0,\lock(m_g),v_{1})}) (t_0,t_1)$
% \[
% \begin{array}{lllr}
%     \semT{\bar a_i}(t_0,t_1)
%     &=&
%     (\semT{(v^i_{2},\unlock(m_g),v^i_{3})} \circ \semT{(v^i_{1},a_i,v^i_{2})} \circ \semT{(v^i_0,\lock(m_g),v^i_{1})}) (t_0,t_1)
% \end{array}
% \]
with composition lifted to return $\emptyset$ when an intermediate computation does so.

\begin{definition}[Bidirectional Compatibility]\label{def:bidirectional-compatibility}
    We call two access sequences $\bar \chi_a$ and $\bar \chi_b$ for a global $g$ bidirectionally (trace)
    compatible if
    \[
    \begin{array}{lll}
        \exists t_a, t_b, t_l \in \Traces:
        \left(\semT{\bar \chi_b}(t_b,\semT{\bar \chi_a}(t_a,t_l)) \neq \emptyset\right) \land
        \left(\semT{\bar \chi_a}(t_a,\semT{\bar \chi_b}(t_b,t_l) \neq \emptyset)\right),
    \end{array}
    \]
    i.e., both orders of accesses are possible.% of the accesses yield local traces.
\end{definition}
Here, $t_l$ corresponds to the local trace ending in the last unlock of $m_g$ before the pair of
considered accesses (or the initialization if no accesses have happened yet), whereas $t_a$ corresponds to the local
trace ending in the last action of the thread performing $\bar \chi_a$ before the pair of considered accesses,
and conversely for $t_b$ and $\bar \chi_b$.
It thus remains to relate the two definitions:
\begin{theorem}
    Two access sequences for a global variable $g$---at least one of which corresponds to a write---are bidirectionally compatible if and only if the accesses are racy.
\end{theorem}
\begin{proof}
    We consider the two directions separately:
    \begin{itemize}
        \item[($\Rightarrow$)] Consider two bidirectionally compatible access sequences $\bar \chi_a$ and $\bar \chi_b$,
            with at least one a write.
            Then, $\exists t_a, t_b, t_l \in \Traces:
            \semT{\bar \chi_b}(t_b,\semT{\bar \chi_a}(t_a,t_l)) \neq \emptyset$. Let us call the trace in this set $t'$.
            By the construction of the actions for locking atomicity mutexes, the two accesses in $t'$ are unordered
            safe for the $m_g$ mutexes and any program order edges (as these are the only edges introduced).
            For the accesses to be ordered by program order edges, both access would have to belong to the same thread,
            as program order edges are thread-local.
            However, in this case we would have $\semT{\bar \chi_a}(t_a,\semT{\bar \chi_b}(t_b,t_l)) = \emptyset$, as
            the access $\bar \chi_a$ would already appear in $t_b$.
            Thus, the accesses are unordered safe for $m_g$ mutexes and thus racy.
        \item[($\Leftarrow$)]
            Consider two racy accesses $\chi_a$ and $\chi_b$ of $g$. At least one is
            a write and there exists a local trace $t$ in which these
            accesses are unordered safe for the $m_g$ mutexes.
            Consider the subtrace ending in the second access.
            The accesses in this subtrace remain unordered safe for the $m_g$ mutexes. We prolong this trace by
            appending the unlock operation of $m_g$. This yields another trace, say $t'$, in which the accesses
            remain unordered safe for the $m_g$ mutexes.
            $t'$ now contains sequences $\bar \chi_a, \bar \chi_b$ corresponding to accesses to $g$ and
            locking and unlocking $m_g$.

            Then, let the subtrace of $t'$ ending in the
            last program point preceding the access sequence of the thread performing the first access be called
            $t_a$ and conversely for $t_b$ and the second access. Furthermore, let $t_l$ denote the subtrace ending
            in the last $\unlock(m_g)$ before either of these access sequences or in $\init(m_g)$ if no such unlock exists.
            Then,
            $\semT{\bar \chi_b}(t_b,\semT{\bar \chi_a}(t_a,t_l)) = t'$ and
            $\semT{\bar \chi_a}(t_a,\semT{\bar \chi_b}(t_b,t_l))$ also yields a local trace,
            as both $t_a$ and $t_b$ originate from the same computation, $\semT{\bar \chi_b}(t_b,t_l)$
            ends in $\unlock(m_g)$, and there are no requirements enforcing a particular order of $\bar \chi_a$ and $\bar \chi_b$.
            Thus, $\bar \chi_b$ and $\bar \chi_a$ are bidirectionally compatible. $\hfill\qed$
    \end{itemize}
\end{proof}
Thus, both definitions coincide.
\begin{figure*}
    \ltDrawingDefs
        \centering
        \begin{tikzpicture}
            \node[programpointwide](tmpp1){};

            \programon{tmpp1}{\init(m_g)}{tmpp2}
            \programon[wide]{tmpp2}{\create(u8)}{tmpp3}

            \programon[wide]{tmpp3}{\lock(m_g)}{tmpp4}
            \programon{tmpp4}{g=3}{tmpp5}
            \programon[wide]{tmpp5}{\unlock(m_g)}{tmpp6}

            \node[programpointnarrow,below left=of tmpp3,draw=none](t1pp0){};
            \node[programpointwide,right of=t1pp0](t1pp1){};

            \programon[wide]{t1pp1}{\lock(m_g)}{t1pp2}
            \programon{t1pp2}{g=5}{t1pp3}
            \programon[widesink]{t1pp3}{\unlock(m_g)}{t1pp5}

            \mutexo[below=1mm]{tmpp6}{t1pp2}{m_g}{out=330,in=110}
            \mutexo[below=1mm]{tmpp2}{tmpp4}{m_g}{out=340,in=200}
            \createo[below=2mm]{tmpp2}{t1pp1}

        \end{tikzpicture}
    \caption{An example local trace of a racy program.}\label{f:prog0lt}
\end{figure*}
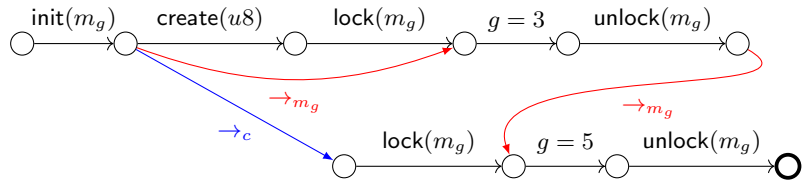
\begin{remark}
    Unlike in many other definitions, the local trace semantics assigns meaning even to executions past their first racy access, assuming sequential consistency.
    Thus, the local trace semantics has a notion of traces containing multiple races.
    However, our notion of how program execution continues \emph{after the first race} need not coincide with other notions:
    For instance, \emph{garbled} writes are not considered in the local trace semantics.
    A semantics that allows for such garbled writes may then disagree with the local trace semantics on races that happen after the first race.
\end{remark}

\section{Races in Terms of Digests}\label{s:races-digest}
With the previous section having established bidirectional trace compatibility as a key ingredient to a formal definition
of data races in the local trace semantics, we now turn to the question of designing sound abstractions of
bidirectional trace compatibility that can be used in practical analyses.
To this end, we introduce for each digest a family of new (commutative) predicates
\[
    (||^{?}_g): \Digests \to \Digests \to \{ \textsf{false}, \top  \}
\]
which is meant to express whether two accesses to global $g$ associated with the respective digests can happen in parallel.
Here, $\textsf{false}$ means that two accesses with these digests can definitely not happen in parallel,
whereas $\top$ means that this cannot be excluded.\footnote{As the digests are an overapproximation, they
do not generally carry enough information to determine that two accesses can definitely happen in parallel.}

We tie this definition
directly to the definition of bidirectional trace compatibility (\cref{def:bidirectional-compatibility}).
We call $||^{?}_g$ \emph{sound} if for all traces $t_a$, $t_b$, $t_l$ and sequences
$\bar \chi_a$ and $\bar \chi_b$  of accesses to some global $g$ as above, the following implication holds:
\begin{equation}
    \begin{array}{lll}
        \emptyset \neq \semT{\bar \chi_b}(t_b,\semT{\bar \chi_a}(t_a,t_l))
        \land  \emptyset \neq \semT{\bar \chi_a}(t_a,\semT{\bar \chi_b}(t_b,t_l))\\
        \qquad\qquad\implies (\alpha_{\Digests}\,t_a)\;||^{?}_g\;(\alpha_{\Digests}\,t_b) = \top
    \end{array}
    \label{eq:mhp-sound}
\end{equation}

\subsection{Bidirectional Digest Compatibility}
A natural question arises:
Can we soundly derive a predicate $||^{?}_g$ for any admissible, access stable digest?
To this end, we first define the notion of \emph{bidirectional digest compatibility}:
\begin{definition}
    Sequences $\bar \chi_a$ and $\bar \chi_b$ corresponding to accesses to a global $g$ are called \emph{bidirectionally digest
    compatible} if
    \begin{equation}
    \begin{array}{lll}
        \exists A_a, A_b, A_l \in \Digests:\\
        \quad \left(\aSemDigest{\bar \chi_b}(A_b,\aSemDigest{\bar \chi_a}(A_a,A_l)) \neq \emptyset\right) \land
        \left(\aSemDigest{\bar \chi_a}(A_a,\aSemDigest{\bar \chi_b}(A_b,A_l) \neq \emptyset)\right).
    \end{array}
            \label{eq:bidirec-race}
    \end{equation}
\end{definition}
\begin{theorem}
    Bidirectional trace compatibility implies bidirectional digest compatibility.
\end{theorem}
\begin{proof}
    By repeated application of the properties of digests (\cref{def:ADet,def:ASound}), and $\alpha_{\Digests}$ being total. $\hfill\qed$
\end{proof}
We remark that the opposite direction does not hold: \eqref{eq:bidirec-race} holding for two sequences $\bar \chi_a$ and
$\bar \chi_b$ does not imply that these sequences are in fact bidirectionally trace compatible, as the digests (by design)
introduce an overapproximation.

With this insight, we can now give a definition of the predicate $||^{?}_g$ for admissible and \emph{access stable} digests.
\begin{proposition}\label{prop:generic-mhp}
For an access stable and admissible digest $\Digests$, the definition
\[
    A_a\;||^{?}_g\;A_b = \begin{cases}
        \textsf{false} & \text{if } \aSemDigest{\lock(m_g)}(A_a,A_b) = \emptyset
         \lor \aSemDigest{\lock(m_g)}(A_b,A_a) = \emptyset\\
        \top & \text{otherwise}
    \end{cases}
\]
is sound.
\end{proposition}
\begin{proof}
This amounts to showing \cref{eq:mhp-sound}.
Consider two bidirectionally compatible access sequences $\bar \chi_a$ and $\bar \chi_b$ and traces $t_a$, $t_b$, and $t_l$ such that
\[
    \begin{array}{lll}
        \left(\semT{\bar \chi_b}(t_b,\semT{\bar \chi_a}(t_a,t_l)) \neq \emptyset\right) \land
        \left(\semT{\bar \chi_a}(t_a,\semT{\bar \chi_b}(t_b,t_l) \neq \emptyset)\right).
    \end{array}
\]
Let $A_a = \alpha_{\Digests}\,t_a$, $A_b = \alpha_{\Digests}\,t_b$, and $A_l = \alpha_{\Digests}\,t_l$.
By the properties of digests (\cref{def:ADet,def:ASound}), also
\[
\begin{array}{lll}
    &\emptyset \neq
    \aSemDigest{\bar \chi_b}(A_b,\aSemDigest{\bar \chi_a}(A_a,A_l))
    \land
    \emptyset \neq \aSemDigest{\bar \chi_a}(A_a,\aSemDigest{\bar \chi_b}(A_b, A_l)).
\end{array}
\]
As the digest is \emph{access stable} (\cref{def:Astable}), this implies
$\emptyset \neq
    \aSemDigest{\bar \chi_b}(A_b, A_a)
    \land
    \emptyset \neq \aSemDigest{\bar \chi_a}(A_a, A_b)
$
which in turn implies
$\emptyset \neq \aSemDigest{\lock(m_g)}(A_b, A_a) \land
\emptyset \neq \aSemDigest{\lock(m_g)}(A_a, A_b)$,
and thus $A_a \;||^{?}_g\; A_b = \top$.
$\hfill\qed$
\end{proof}

\begin{example}
    Consider again the program in \cref{l:prog1}, which is annotated using the digest for lightweight thread \emph{id}s (\cref{f:threadingmainA}).
    For the two access sequences starting after line \ref{line:mg:firstAccessSequenceInMain} and line \ref{line:mg:accessSequenceInT1}, with associated digests $\STmain$ and $\MT$, it can be excluded that
    these two accesses race with each other: While
    $\semDigests{\lock(m_g)}\,(\MT,\STmain) = \{\MT\} \neq \emptyset$,
    $\semDigests{\lock(m_g)}\,(\STmain,\MT) = \emptyset$ and the respective sequences thus
    are known to not be bidirectionally compatible.
    $\hfill\qed$
\end{example}
This definition in \cref{prop:generic-mhp} is sound and generic---but can be overly conservative:

\begin{figure*}[t]
    \begin{minipage}[t]{.45\linewidth}\[
        \begin{array}{lll}
            \Digests = 2^\Mutexes\\
            \initDigests = \{\emptyset\}\\
            \newDigests\,S\,u_1 = \{\emptyset\}\\
            \semDigests{\act}\,S = \{S\} \; \text{(other non-observing)}\\
            \semDigests{\act}\,(S_0,S_1) = \{S_0\} \; \text{(other observing)}
        \end{array}
    \]
    \end{minipage}
    \begin{minipage}[t]{.45\linewidth}\[
        \begin{array}{rll}
            \semDigests{\lock(a)}\,(S_0,S_1) &=& \begin{cases}
                \emptyset & \text{if } a \in S_0\\
                \{S_0 \cup \{a\}\} & \text{otherwise}
                \end{cases}\\[1ex]
            \semDigests{\unlock(a)}\,S &=&
                \begin{cases}
                    \emptyset & \text{if } a \not\in S\\
                    \{S \setminus \{a\}\} & \text{otherwise}
                \end{cases}
        \end{array}
        \]
    \end{minipage}
    \caption{Right-hand sides for lockset digest~\cite{hammerandnail}.}\label{f:locksets}
\end{figure*}

\begin{example}\label{ex:lockset-falls-short}
    Consider the lockset digest~\cite{hammerandnail} (see \cref{f:locksets}) which associates with each program point
    the set of locks held when it is reached.
    The program in \cref{l:prog1} is annotated with this information.
    The digests for lines \ref{line:mg:secondAccessSequenceInMain} and line \ref{line:mg:accessSequenceInT1},
    right before the following access sequences, are given by $\{a\}$.
    $\aSemDigest{\lock(m_g)}(\{a\},\{a\})$ yields $\{\{a,m_g\}\}$, i.e., a non-empty set of digests,
    and thus $\{a\}\;||^{?}_g\;\{a\} = \top$.
    Using the generic definition of $||^{?}_g$ from \cref{prop:generic-mhp}, we cannot exclude that the accesses
    in line \ref{line:mg:secondAccessInMain} and line \ref{line:mg:accessInT1} race, despite both happening while holding
    the mutex $a$.$\hfill\qed$
\end{example}

\begin{remark}
    Setting $\aSemDigest{\lock(m_g)}(S_a,S_b)=\emptyset$ when $S_a \cap S_b \neq \emptyset$ renders the digest inadmissible:
    Consider a program where all accesses to a global $g$ are protected by some mutex $a$.
    Then, any unlock of $m_g$ happens while holding $a$, as does any lock of $m_g$.
    Setting $\aSemDigest{\lock(m_g)}(\{a\},\{a\})=\emptyset$ thus violates \cref{def:ASound}.
\end{remark}

Nevertheless, two accesses where a common program mutex is held cannot race.
% if they belong to the same execution.
In fact, synchronization via mutexes is a common way to prevent races.
Thus, this first definition in terms of bidirectional digest compatibility---while sound and generic---does not sufficiently capture the property for all digests.

\subsection{Bespoke Definitions}\label{s:mhp-predicate}
While the definition in \cref{prop:generic-mhp} allows using information computed by any access stable digest
to exclude races, often better definitions of $||^{?}_g$ are possible.
\begin{example}\label{ex:lockset-mhp}
Consider again the definition of $||^{?}$ for the lockset digest
provided in \cref{ex:lockset-intro}.
\[
    S_a\;||^{?}\;S_b = \begin{cases}
        \textsf{false} & \text{ if } S_a \cap S_b \neq \emptyset\\
        \top & \text{otherwise}
    \end{cases}
\]
This definition is more precise than the definition in terms of bidirectional digest compatibility
and captures the intent behind mutex-based synchronization.
$\hfill\qed$
\end{example}
\begin{proposition}
    The definition of $||^?$ for the lockset digest given above is sound.
\end{proposition}
\begin{proofsketch}
    While immediately believable, we still sketch a proof.
    It is by contraposition: Consider sequences $\bar \chi_a$ and $\bar \chi_b$ of accesses to a global $g$
    and local traces $t_a$, $t_b$ and $t_l$ where at the end of $t_a$ and $t_b$,
    both threads hold some mutex $m$. Then,
    $ m \in (\alpha_{\Digests}\,t_a) \cap (\alpha_{\Digests}\,t_b)$ and thus
    $(\alpha_{\Digests}\,t_a)\;||^{?}\;(\alpha_{\Digests}\,t_b) = \textsf{false}$.
    We show that then the left-hand side of implication \eqref{eq:mhp-sound} is false, i.e.,
    the accesses are not bidirectionally compatible.
    If both accesses are performed by the same thread, these sequences are not bidirectionally compatible as the program
    order orders one sequence before the other.
    Consider the case where the accesses are performed by different threads: Neither $\bar \chi_a$ nor $\bar \chi_b$ unlocks $m$,
    and both $t_a$ and $t_b$ contain operations that lock $m$.
    As locks of $m$ are totally ordered, one of the sequences must have happened before the
    other thread acquired $m$.
    Thus, trace compatibility holds at most in one direction.
    For an in-depth proof, see
    \ifdefined\extended
\cref{app:lockset-mhp-sound}.
\else
    the extended version~\cite[Appendix A]{extended-version}.
\fi
    $\hfill\qed$
\end{proofsketch}
Later, we will usually not flesh out soundness arguments where they are intuitive.

\begin{example}\label{ex:threadingextra}
    For a further example, consider again the digest from \cref{f:threadingmainA}.
    The definition in terms of bidirectional digest compatibility already serves to exclude races
    where one access happens in single-threaded mode whereas the other
    one happens in multi-threaded mode.
    However, two accesses performed by the main thread can also not race as the main thread is unique,
    and such accesses are totally ordered by the program order.
    To exploit this, we can set:
    \[
    M_a\;||^{?}\;M_b = \begin{cases}
        \textsf{false} & \text{ if } ( M_a = \STmain \lor M_b = \STmain ) \lor (M_a = M_b = \MTmain)\\
        \top & \text{otherwise}.
    \end{cases}
\]
% to beat the definition in terms of bidirectional digest compatibility.
\ifdefined\extended
\Cref{s:thread-ids} provides a definition for general thread \emph{id}s.    $\hfill\qed$
\else
    We define $||^?$ for general thread \emph{id}s in the extended version~\cite[Appendix B]{extended-version}. $\hfill\qed$
\fi
\end{example}
\ignore{
The same holds for a thread join analysis where an unlock by a thread definitely not running anymore can be the last unlock,
but the respective sequences are never bidirectionally compatible.
Also, for a thread which has been determined to be unique, while bidirectional digest compatibility is often true,
each action taken by the thread is totally ordered w.r.t.\ all the other actions of this thread, and bidirectional
trace compatibility does not hold. An appropriate definition of the predicate $||^{?}$ can capture this.}
\subsection{Combination of Digests}
\begin{figure}[t]
\begin{minted}[escapeinside=@@,fontsize=\scriptsize,mathescape=false,linenos]{python}
# From program, digest, and RHS of abstract interpretation, construct constr. sys. (Sect. @\ref{s:digests}@)
constraintSystem = constructConSys(program, digest, (@$\llbracket\cdot\rrbracket^{\sharp},\new^\sharp,\init^\sharp,\D$@)); @\label{line:constr:sys}@

# Add helper unknowns and constraints (@$\cref{eq:cons-helper-accesses}$@) and solve using fixpoint engine
hConstraintSystem = addHelperConstraints(constraintSystem); @\label{line:helper:constrs}@
solution = solve(hConstraintSystem); @\label{line:solution}@

# Postprocessing
for g in globals: @\label{line:global:loop}@
    accesses = solution[@${\cal R}_g$@];  # Get set of accesses to g @\label{line:accesses}@

    # For each pair of accesses, check whether they may race and at least one is a write
    for (@$u_0$@, @$\tau_0$@, @$A_0$@) in accesses: @\label{line:outer:access:loop}@
        for (@$u_1$@, @$\tau_1$@, @$A_1$@) in accesses: @\label{line:inner:access:loop}@
            if (@$\tau_0$@ == W or @$\tau_1$@ == W ) and (@$A_0 ||^?_g A_1$@ == @$\top$@): @\label{line:check:race}@
                report_race(@$u_0$@, @$u_1$@); @\label{line:report:race}@
\end{minted}
\caption{Pseudocode of the digest-driven algorithm for data race detection.}
\label{f:algo}
\end{figure} %

\noindent Several admissible digests can be combined into an admissible product digest~\cite{hammerandnail} with all operations
given pointwise. If all digests are access stable, so is the product digest. To leverage such a product digest for
race detection, we consider the range of answers
$\{\textsf{false}, \top\}$ of the predicate $||^{?}_g$ as a lattice, with $\textsf{false}$ as the least element.
The predicate $||^{?}_g$ for a product
$(\Pi_{j=0}^{n-1}\,\Digests_j)$
of $n$ \emph{active digests} then is defined as the component-wise
meet of predicates $||^{?}_{g,j}$ for individual digests
\[
    (A^0_0,\dots,A^0_{n-1})\;||^?_g\;(A^1_0,\dots,A^1_{n-1}) = \textstyle{\bigsqcap_{j=0}^{n-1}} (A^0_j\;||^?_{g,j} \;A^1_j)
\]
As all $||^{?}_{g,j}$ soundly overapproximate bidirectional
% trace
compatibility, so does the meet.
\begin{example}
    Consider again the program in \cref{l:prog1} where the main thread first initializes a global $g$
    and then starts two threads which
    each write to $g$, with only the latter accesses synchronized via some mutex.
    Neither the lockset digest (\cref{ex:lockset-mhp})
    nor the lightweight thread \emph{id} digest (\cref{{ex:threadingextra}})
    on their own suffice to show the absence of all data
    races, but their combination does.
$\hfill\qed$
\end{example}
Alternatively, one can define a custom predicate for the product digest to jointly exploit information
from all digests---at the expense of providing a dedicated soundness argument.
Such a custom predicate resembles the \emph{reduced} product construction as proposed by \citet{CousotC79},
while the definition via the component-wise meet is more in the spirit of the \emph{direct} product.
% \begin{remark}
%     In the same manner a predicate $||^?: \Digests \to \Digests \to \{ \textsf{false}, \top  \}$ can be defined for
%     digests,
%     an equivalent predicate can be defined
%     for abstract digests~\cite{hammerandnail}, which are further abstractions of digests.
%     For details, see \cref{app:abstractdigest}.
% \end{remark}

% \begin{remark}\label{r:racecheckarbitraryactions}
%     For those digests that are unaffected by a sequence corresponding to an access, it is also possible to apply
%     the predicate to check whether two arbitrary actions are unordered w.r.t.\ each other. Conceptually, this works
%     by modifying the program to perform write accesses to a fresh global variable while holding only the respective
%     atomicity mutex. If these two accesses race, the original operations are unordered w.r.t.\ each other.
% \end{remark}

\section{The Digest-Driven Race Detection Algorithm}\label{s:algorithm}
The predicate $||^{?}_g$ as a convenient yet expressive overapproximation of bidirectional trace compatibility can
serve as a cornerstone of a race detection algorithm for a static analyzer.
Such an algorithm works in two phases: During the analysis,
for each global $g$ the digest associated with each access and the access type is recorded.
In a post-processing phase, the predicate $||^{?}_g$ is then used to check whether there are any accesses that may race.

The algorithm is shown in \cref{f:algo}.
In line \ref{line:constr:sys} a refined abstract constraint system is constructed that is parameterized by the program,
the abstract domain and the digest, which we require to be admissible and access stable.
%
% based on the refined constraint system from \cref{s:digests}.
%
%
This constraint system is extended with constraints
that enforce that for each global variable information about each access to it is accumulated (line \ref{line:helper:constrs}).
For each global $g$, a dedicated unknown $[\raceUnknown{g}]$ is introduced, with values ranging over
$2^{(\mathcal{N} \times \{\text{W},\text{R}\} \times (\Pi_i \Digests_i))}$
where each tuple corresponds to an access. It consists of the location of access, its type ($\text{W}$(rite) or $\text{R}$(ead)),
and digest information from all active digests.
%
% \todo[inline]{If we go for side-effecting constraints.}
% In an analysis described as a side-effecting constraint system, this collection can be achieved by lifting the right-hand
% sides for unlocking the atomicity mutexes after an access
% \[
%     \begin{array}{l}
%         \aSem{[u,A], \unlock(g),A'}_{\mathcal{R}}\eta =\\
%         \qquad \Let\;(\rho,\sigma) = \aSem{[u,A], \unlock(g),A'}\eta\;\In\\
%         \qquad (\rho \cup \{ [\raceUnknown{g}] \mapsto \{ (u,\tau\,u,A')\}  \}, \sigma)
%     \end{array}
% \]
% where $\tau\,u \in \{W,R\}$ is the type of the access which can be inspected by checking whether the single incoming
% edge of $u$ (this property is guaranteed by construction of the programs with atomicity mutexes) is a write or a read.
% \todo[inline]{If we go for digest-style constraints.}
This is achieved by introducing for each constraint of the form
\[
    \begin{array}{l}
        \relax[ \unlock(m_g), A'] \sqsupseteq \aSem{(u_1,\unlock(m_g),u_2)} [u_1,A_0]
    \end{array}
\]
an additional constraint
\begin{equation}
    \begin{array}{l}
        \relax[\raceUnknown{g}] \sqsupseteq \left(\lambda\;x \to
            \begin{cases}
                \{ (u_1,\tau\,u_1,A')\}  & \textsf{if } x \neq \bot\\
                % Was \aSem{(u_1,\unlock(m_g),u_2)}\,x \neq \bot
                \emptyset & \textsf{else}
            \end{cases} \right) [u_1,A_0]
    \end{array}\label{eq:cons-helper-accesses}
\end{equation}
where $\tau\,u_1 \in \{W,R\}$ is the type of the access which can be determined by checking whether the single
(by construction of the programs with atomicity mutexes) incoming edge of $u_1$ is a write or a read.
The case distinction ensures that only the digests associated with potentially reachable accesses are recorded,
while the digest $A'$ in this case is guaranteed to be identical to the digest before executing the access sequence
due to access stability (\cref{def:Astable}).
The constraint system is then solved (line \ref{line:solution}).
In the post-processing phase (line \ref{line:global:loop} to \ref{line:report:race}) for each global the set of accumulated accesses is considered.
For any two program points $u_0$ and $u_1$, a potential race for a global $g$ is then flagged if
\[
    \begin{array}{l}
    \exists \tau_0,A_0,\tau_1,A_1{:} (u_0,\tau_0,A_0), (u_1,\tau_1,A_1) {\in} [\raceUnknown{g}]\land
    (\tau_0 {=} W {\lor} \tau_1 {=} W) \land (A_0\;||^?_g\;A_1) {=} \top,
    \end{array}
\]
i.e., at least one of the accesses is a write, and it cannot be excluded that they are unordered w.r.t.\ each other.
We remark that the tuples $(u_0,\tau_0,A_0)$ and $(u_1,\tau_1,A_1)$ are allowed to coincide,
as two accesses with the same digest may still race as they could be performed by different (concrete) threads.

\begin{remark}
Where an analyzer does not have concrete memory locations, but only abstract memory locations,
as may be the case
when analyzing programs with dynamic memory allocation with an allocation site abstraction, the
check is performed for all accesses to abstract memory locations that potentially overlap.
\end{remark}

% We posit that digests may also be key to enabling the analyzer to provide the user with succinct explanations, e.g.,
% detailing why an access was deemed to be non-racy (choose most understandable
% digest to explain this), or why an access could not be proven racy by giving conditions under which the desired property could be proven.
% Such an explanation may, e.g., a thread with thread \emph{id} $i_0$ is known to be must joined at this point, but its uniqueness could not be established.
% %
% Smartly combining different digests may also help explain near-misses: If, e.g., the may lockset digest contains one specific mutex in all accesses to a
% given global, one may hypothesize that
% this mutex is meant to guarantee mutual exclusion and it may be helpful to generate an explanation why it is not known to be held.
% Similarly, if all accesses are protected by a mutex and this is also the only reason all other accesses are race-free one may claim
% that in this location the locking of this mutex has been forgotten.

\section{Experiments and Benchmark Set}\label{s:experiments}

The \Goblint static analyzer~\cite{VojdaniARSVV16} is based on thread-modular abstract interpretation of the local trace
semantics~\cite{schwarz2023clustered,schwarzthesis}.
While it has supported data race detection for a long time, digest-driven data race detection proposed here serves as the
theoretical underpinning and justification for this feature. \Goblint implements various
digests, and the digest-driven algorithm we propose in \cref{s:algorithm} is what emerged after unifying various older ad-hoc mechanisms.
% previously used in \Goblint.

%
\Goblint participates in the data race track of % the International Competition on Software Verification
\SVCOMP since its inception % as a demo category
in 2021, receiving the highest score in 2021, 2023, and 2024\footnote{\textsc{Deagle} initially scored higher but was
disqualified for fingerprinting tasks~\cite{Beyer24}.} and the second-highest score in
2025---beaten only by \textsc{CoOperRace} which internally calls \Goblint.~\cite{Beyer21,Beyer23,Beyer24,BeyerS25}
\cref{table:svcomp-25-no-data-race-main} shows the 2025 results.
% of tools producing at
% least 650 correct true verdicts, i.e., programs where the tool correctly deduced race-freedom. % concluded that the given task contains no data races.
Among these, \Goblint\ is the only tool producing no incorrect verdicts.
\begin{table}[t]
	\centering
	\footnotesize
\begin{tabular}{l *{5}{r}}
\toprule
& \textsc{CoOpeRace} & \textsc{Goblint} & \textsc{Infer} & \textsc{Deagle} & \textsc{RacerF} \\
\midrule
Correct true      & 754 & 712 & 747 & 685 & 674 \\
Correct false     &   6 &   0 &   0 &   0 &  68 \\
Incorrect true    &   1 &   0 & 213 &   1 &   0 \\
Incorrect false   &   0 &   0 &   0 &   0 &   4 \\
\bottomrule
\end{tabular}
	\vspace{1em}
	\caption{SV-COMP 2025 results for no-data-race~\cite{svcomp25NoDataRaceMain} of tools with at least 650 correct true verdicts.
	% excluding unconfirmed results.
	\emph{Correct true} (\emph{correct false}) is the number of programs where a tool correctly deduced race-freedom (raciness).
	\emph{Incorrect true} (\emph{incorrect false}) is the number of programs where a tool incorrectly claimed race-freedom (raciness).\label{table:svcomp-25-no-data-race-main}
	}
\end{table}
%

%
%
% For this evaluation, we
We pose the following research question about digest-driven data race detection as implemented in \Goblint:
\begin{description}
    \item[RQ] What impact do different digests and combinations of digests have on the number of solved tasks in the \SVCOMP\ data race benchmark set?
\end{description}
We modify
\Goblint\ to allow us to selectively change the predicate $||^?_g$ for some digests to always return $\top$.
Then, we compare for how many tasks \Goblint\ still succeeds in proving data-race freedom depending on which
digests are used to exclude races.
\Goblint\ implements the following digests:\footnote{
Digests marked with a $^\star$ go beyond the framework presented here in that they
% collect auxiliary information about the memory locations and
combine  time and space aspects.
We include them as they nevertheless give rise to a predicate $||^{?}$ akin to ours and can thus be employed in the digest-based algorithm.
}
\begin{description}[labelwidth=0.8cm,leftmargin=1.2cm]
    \item[L] Locksets (c.f. \cref{f:locksets} and \cite{hammerandnail}, plus support for Reader/Writer-Locks)
    \item[TF] Threadflag for multi- vs. single-threaded mode (c.f. \cref{f:threadingmainA})
    \item[TID] History-based thread \emph{id}s~\cite{schwarz2023clustered} (recording parts of the creation history)
    \item[J] Must-joined threads~\cite{schwarz2023clustered} (based on the history-based thread \emph{id}s)
    \item[FR] Tracking allocated memory that has not escaped its thread, i.e., is \emph{fresh}\footnote{
        It may seem surprising that FR can be considered a time-separation property.
        However, any access where the accessed memory has not escaped its thread yet
        must be ordered w.r.t. all other accesses to the same memory: Either both accesses happen in the same thread, or the other access
        happens later after the memory has escaped.
    }
    \item[SL] Symbolic per-element locksets~\cite[Sections 5\&6]{VojdaniARSVV16}$^\star$
    \item[R] Region analysis~\cite{seidl2009region}$^\star$
\end{description}
The benchmarks were conducted on a machine with two \textsf{Intel Xeon Platinum 8260 @ 2.40GHz}
processors with 24 cores each and 256GB of RAM running \textsf{Ubuntu 24.04.1 LTS}.
\textsc{Benchexec}~\cite{BeyerLW19} was used to limit each run to 5 minutes and 5GB of memory.
With these settings, the system ran out of memory for 4 tasks, 3 tasks timed out, and the
analyzer crashed for 7 tasks. Where execution terminated, it did so within at most 20s, where for
all but 6 tasks the runtime was below 10s. As the analysis was always performed with the same settings,
and only the definition of the predicate $||^?_g$ was changed, we do not further detail runtimes.\footnote{
For \textbf{TID}, prior experiments~\cite[Chapter 5.3.2]{schwarzthesis} on large, real-world programs,
have established a median slowdown of around $1.2\times$ (with a maximum slowdown of $5.2\times$) in the context of value analyses.}
\begin{figure}[t]
    \newcommand{\on}{$\bullet$}
    % Define a command for creating mini bars
    \newcommand{\vbar}[1]{%
        \edef\computedheight{\fpeval{#1/794*3}}% Scale to max height 3
        \begin{tikzpicture}[baseline=0.5ex]
            \fill[black] (0,0) rectangle (0.3,\computedheight);
        \end{tikzpicture}%
    }

    \newcommand{\axis}{%

        \begin{tikzpicture}[baseline=0.5ex]
            % Draw axis line
            \draw (0,0) -- (0,3) node {};
            % Add ticks and labels
            \foreach \y/\label in {0/0, 1.13/300, 2.27/600, 3/794} {
                \draw (0,\y) -- (-0.1,\y) node[left, xshift=-2pt] {\tiny \label};
            }
        \end{tikzpicture}%
    }

    \centering
    \newcolumntype{C}{>{\centering\arraybackslash}p{0.5cm}}
    \newcolumntype{G}{>{\columncolor{gray!20}\centering\arraybackslash}p{0.5cm}}

    \begin{tabular}{|l|GCGCGCGCGCGCGC|}
    \toprule
    \textbf{L}   &    & \on &     &       &       &       &         &     & \on  &  \on  & \on  & \on &  \on   &  \on  \\
    \textbf{TF}  &    &     & \on &       &       &       &         &     & \on  &       &      & \on &  \on   &  \on  \\
    \textbf{TID} &    &     &     &  \on  &       &       &         & \on &      &  \on  & \on  & \on &  \on   &  \on  \\
    \textbf{J}   &    &     &     &       &       &       &         & \on &      &       & \on  & \on &  \on   &  \on  \\
    \textbf{FR}  &    &     &     &       &  \on  &       &         &     &      &       &      & \on &  \on   &  \on  \\
    \textbf{SL}  &    &     &     &       &       &  \on  &         &     &      &       &      &     &  \on   &  \on  \\
    \textbf{R}   &    &     &     &       &       &       &  \on    &     &      &       &      &     &        &  \on  \\
    \midrule
    \textbf{\#}  & 15 & 128 & 52 &  91  &  15  &  29  &  16  & 166 & 492 & 535 & 681   & 681   &  693    &  712  \\
    \bottomrule

    \axis & \vbar{15} & \vbar{128} & \vbar{52} & \vbar{91} & \vbar{15} & \vbar{29} & \vbar{16} & \vbar{166} & \vbar{492}  & \vbar{535} & \vbar{681} & \vbar{681} & \vbar{693} & \vbar{712}  \\
    \bottomrule
    \end{tabular}

    \caption{The number of tasks where each digest succeeded in proving race freedom, out of a total of
    794 race-free tasks in the benchmark set. None of the approaches produced any false negatives for the other
    235 tasks in the benchmark set, which are racy.
    }
    \label{f:results}
\end{figure}
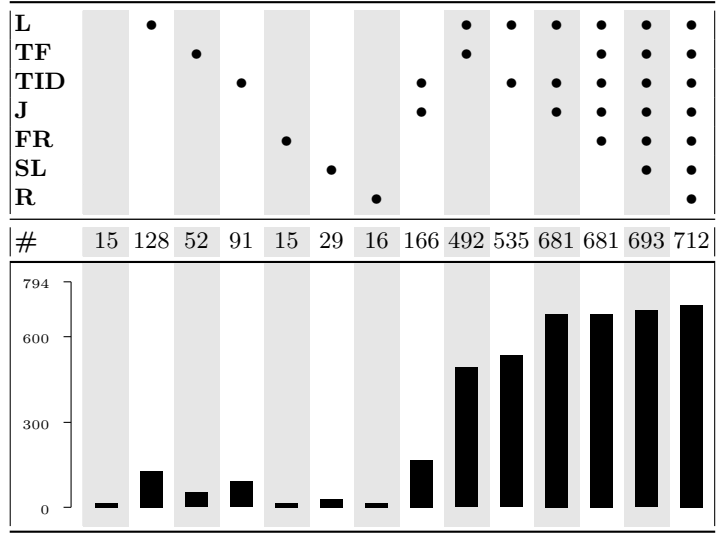
\Cref{f:results} summarizes the evaluation results: Each column corresponds to one configuration,
and a $\bullet$ in a row indicates that this configuration uses the given digest to exclude races.
The number indicates for how many out of a total of 794 race-free tasks the configuration succeeded in establishing
race freedom.

Interestingly, for $15$ of the $794$ tasks, none of the digests
is necessary to prove race freedom. These programs only contain accesses to shared memory locations
marked as atomic using the \texttt{\_Atomic} specifier from C11 or its \textsc{GCC} equivalent.
When it comes to single digests, \textbf{L} succeeds in proving race freedom for a sizable chunk of tasks, as do
\textbf{TF} and \textbf{TID}. Here, \textbf{TID} subsumes \textbf{TF}.
The other digests are less successful on their own, succeeding for less than $30$ tasks each.
While enabling \textbf{J} on top of \textbf{TID} almost doubles the number of proven tasks, it is only the combination
of thread-based and lock-based techniques that enables race-freedom to be proven for the majority of tasks
(\textbf{L+TF}: 492, \textbf{L+TID}: 535, \textbf{L+TID+J}: 681). Enabling \textbf{FR} and \textbf{TF}
on top of \textbf{L+TID+J} does not yield an improvement, but enabling \textbf{SL} and \textbf{R} modestly increases the
number of successful tasks to $712$.
Overall, the experiments show that digest-based data race detection and the \emph{combination} of digests reasoning about mutexes, created threads, and joins underlie the success of \Goblint.
The impact of the more specialized techniques used by \Goblint on the other hand is rather limited.

\paragraph{Threats to Validity.}
All digests remained enabled, with only $||^?_g$ modified to always return $\top$.
This may lead to underestimating their impact, as a digest may help show program points unreachable.
However, we do not expect programs to contain many unreachable accesses.
Additionally, digests can constrain other digests with which program points are reached.
As our digests do not exploit guards, we expect such improvements
(through the value analysis) to be rare. % occurrence.

\section{Related Work}\label{s:related}
Detecting data races is a well-studied program analysis problem and
there is a wealth of dynamic
\cite{OCallahanC03,RamanZSVY12,SavageBNSA97,Effinger-DeanLCGB12,MarinoMN09,SerebryanyI09}
and static race detectors \cite{PratikakisFH06, Terauchi08, VojdaniARSVV16, HeSF22, GavrilenkoLFHM19, FarzanKP22,DietschHKSP23,EnglerA03,LiuLLTS021, VoungJL07}.
% Here, we focus on related work on
%
A recent comparison of sound static race detectors~\cite{Holter25} concludes %a survey,
that, to this day, open challenges remain---both w.r.t.\ showing time- and space-separation of accesses.
Heeding this call for further research, our work provides a solid foundation for reasoning about time separation of accesses.
% and proposes novel digests for often-overlooked synchronization primitives.
% \citet{LiuLLTS021} include the thread that allocated a memory object as part of the abstraction to increase the precision of data race detection w.r.t. the space separation of accesses. %May also be just mentioned in cluster cite. They claim however to be very successful.

For the unsound static race detector \textsc{RacerD}, \citet{GorogiannisOS19} state under which (idealized) conditions the tool becomes sound.
Their definition of a race is for an interleaving semantics, and they require that there is an interleaving
that can be extended by performing either of the racy accesses next.
This differs from our setting building on the local trace semantics, where,
in general, no local view of the execution can be extended directly by either racy access, rendering their definition inapplicable for local traces.
\citet{ChopraPD19}, for an analysis of interrupt-driven kernels, define data races based on overlapping sections of executions
which also does not apply to the local trace semantics.
% Unlike our definition this allows also capturing high-level races, but this definition is not applicable to the local trace semantics.
% \paragraph{Dynamic Data Race Detection} For dynamic data race detection, \citet{OCallahanC03} propose an approach that takes information about both held locks and happens-before information into account.\todo{Why this specific dynamic approach? Either we need to cite the dynamic things more widely, or have a cluster cite somewhere...}

A line of work closely tied to the time-separation aspect of data race detection % captured by the $||^?$ predicate
is \emph{May-Happen-in-Parallel} (MHP) analysis~\cite{Naumovich1999,Barik2006,Agarwal07,Di2015,Zhou2018,Albert12,Albert15,Albert17,Lu2025}, which
attempts to determine sets of statements that may execute in parallel.
Approaches there are, once more, often based on flavors of interleaving semantics, rather than a local view.
Typically, control flow within threads,
i.e., whether given statements are reachable at all, is abstracted away.
Approaches may, e.g., extract a graph-based overapproximation of the \emph{inter}-thread control flow in a first step,
with an analysis of this graph as a second step.~\cite{Di2015,Lu2025,Barik2006}
% \citet{Barik2006} presents a more general approach that takes non-unique threads into account.
% This approach relies on a global Thread Creation Tree, computed using symbolic execution, which beyond the thread creation hierarchy
% determines for each thread, even non-unique threads, whether they are definitely joined into the creating thread.
% Based on this structure, one may first identify whether two threads may run in parallel by checking the creation sites of their ancestors within their youngest common ancestor to potentially establish that one is fully joined before the start of the other.
% Then, one may compute MHP information within the nodes of a thread by checking, as we do, if a thread as been fully must-joined before reaching a given node.
% Finally, nodes with overlapping locksets are ruled out.
Our framework, on the other hand, is integrated into a thread-modular abstract interpreter and thus
seamlessly combines MHP reasoning for races with value analyses.
Furthermore, as opposed to the more monolithic approaches in the MHP literature, our framework can be instantiated with different digests
and their admissibility and the soundness of their predicate $||^?_g$ can be established modularly for each digest.

Digests were proposed as a way to refine the static analysis of the \emph{values}
of global variables of multi-threaded programs~\cite{schwarz2023clustered,schwarzthesis} with later
work~\cite{hammerandnail} proposing further instances of digests and generalizing the setup along various axes---within the scope of value analysis.
We do not study value analyses, but instead show how digests can be repurposed
to obtain principled static data race checkers in an extensible framework allowing for modular soundness proofs.

% \paragraph{Concurrency-Sensitivity for Static Data Race Detection.}
% Previous work on con\-curr\-ency-sensitivity \cite{hammerandnail} proposes the digest framework.
% Here we establish that digests as concise representations of the computational past can serve to exclude data races.
% We introduce a dedicated predicate to check whether two accesses may happen in parallel given their digests.
%
% The framework for concurrency-sensitive data race detection presented here casts may capture many seemingly different techniques.
% By providing the requirements for admissibility of a digest, this framework allows for modular soundness proofs.
% Data race detection is a long existing field. We don't believe that this work revolutionizes the field, but instead it
% casts many seemingly different techniques into a common framework with a sound basis, namely the local trace semantics.

% Next, we discuss work that is related to the specific digests we propose.
% \citet{KahlonSKZ09} propose an approach to  static data race detection that excludes some impossible data races where it is known that a thread has already been joined.
% The data race detection tool proposed by \citet{DoscheAC24} detects potential race conditions in the Linux kernel. It considers synchronization via mutexes as well as reader-writer locks.

\section{Conclusion and Future Work}\label{s:conclusion}
We have provided a framework for reasoning about which accesses cannot happen simultaneously for digest-driven abstract interpretation.
Giving two definitions of data races for the local trace semantics and establishing their equivalence
allowed us to derive a sound digest-based data race detection algorithm.
It relies on a novel predicate $||^?_g$ to check whether two accesses with given digests may happen in parallel.
We have shown how a generic definition of this predicate % in terms of bidirectional digest compatibility
can be derived in terms of building blocks provided by any digest, and how, for some digests,
bespoke definitions can yield more precise results.
Digest-Driven Abstract Interpretation underlies our implementation in the static analyzer \Goblint.
To evaluate the impact of digests on the results of \Goblint, we have turned specific digests on or off, finding that it is the combination of digests that underlies its success in \SVCOMP.
Combining the lockset digest with digests reasoning on thread \emph{id}s and joins increases the number of successful tasks
by more than fourfold when compared to the thread-related digests and by more than fivefold when compared to the lockset digest alone.

There are many directions for future work: Preliminary experiments show that digests can capture further synchronization
primitives such as barriers, and that our techniques also apply to the analysis of other concurrency issues such as deadlocks.
Going beyond time separation, it is also worth exploring how to generalize digests to also tackle space separation,
by, e.g., enhancing pointers with allocation time stamps.
Lastly, it would be interesting to investigate how digests can serve as succinct explanations of analysis results to end users
that do not require them to understand the inner workings of the analyzer and are thus easily digestible.

\ignore{For instance listing the sets of held locks at racy accesses may provide strong hints at lapses in the locking discipline,
while the presence of a data race between two accesses where one of the threads should already be joined,
may be an easy-to-understand case of analyzer imprecision. Both should be easily digestible explanations for users of static analyzers.}

\ignore{A first step in this direction could, e.g., be enhancing pointers with allocation time stamps.
In this way, even when dynamically allocated memory is not tracked precisely, it would be known that pointers to the
same abstract chunk of heap, but with different allocation timestamps, do, in fact, represent disjunct regions of the heap.
Also, it would be interesting to see how digests can help detect other concurrency issues, such as deadlocks.
\todo[inline]{Integrate?}
The computed digests do not only underlie the reasoning performed by static analyzers but also serve as an easy-to-understand
insight into the reasoning performed by the analyzer, as it is local to the potential accesses, and often expressible
in a way that does not
require deep insights into the inner workings of the analyzer: For instance listing the sets of held locks at
both accesses may provide strong hints at lapses in the locking discipline, while the presence of a data race
between two accesses where one of the threads should already be joined, may be an easy-to-understand case of analyzer
imprecision.}

\paragraph*{Acknowledgments and Data Availability Statement.}
\ignore{This research is supported in part by the National Research Foundation, Singapore,
and Cyber Security Agency of Singapore under its National Cybersecurity
R\&D Programme (Fuzz Testing NRF-NCR25-Fuzz-0001).
Any opinions, findings and conclusions, or recommendations expressed
in this material are those of the author(s) and do not reflect the views
of National Research Foundation, Singapore, and Cyber Security Agency of Singapore.}
This work is also supported in part by
Deutsche Forschungsgemeinschaft (DFG) --- 378803395/2428
\textsc{ConVeY}.
An artifact~\cite{artifact} allowing for the reproduction of our results is available on Zenodo.

{
  \renewcommand{\doi}[1]{\textsc{doi}: \href{http://dx.doi.org/#1}{\nolinkurl{#1}}}
  \bibliographystyle{splncs04nat}
  \bibliography{bibliography}
}
\ifdefined\extended
\appendix
\begin{section}{Soundness of Predicate $||^?$ for Lockset Digest}\label{s:lockset-mhp-sound}
\label{app:lockset-mhp-sound}
We first recall the proposition from \cref{s:mhp-predicate}:
\begin{proposition}
    The definition of $||^?$ for the lockset digest given above is sound.
\end{proposition}

\begin{proof}
    The proof is by contraposition: Consider two sequences $\bar \chi_a$ and $\bar \chi_b$ of accesses to a global, say $g$,
    and local traces $t_a$, $t_b$ and $t_l$ where at the end of $t_a$ and $t_b$,
    both threads hold some mutex, say $m$. Then, we have
    $ m \in (\alpha_{\Digests}\,t_a) \cap (\alpha_{\Digests}\,t_b)$ and thus
    $(\alpha_{\Digests}\,t_a)\;||^{?}\;(\alpha_{\Digests}\,t_b) = \textsf{false}$.
    We show that then the left-hand side of the implication in \eqref{eq:mhp-sound} does not hold either, i.e.,
    the accesses are not in fact bidirectionally compatible.

    For the left-hand side of the implication to not hold, either $\semT{\bar \chi_b}(t_b,\semT{\bar \chi_a}(t_a,t_l))$
    or $\semT{\bar a\chi_a}(t_a,\semT{\bar \chi_b}(t_b,t_l))$ needs to return $\emptyset$.
    If both return $\emptyset$, nothing is to be shown. Now consider the case where one does not return $\emptyset$.
    We show that in this case, the other one must be $\emptyset$.

    W.l.o.g., assume that $\emptyset \neq \semT{\bar \chi_b}(t_b,\semT{\bar \chi_a}(t_a,t_l))$ holds, i.e., the access
    $\bar a_a$ can happen before the access $\bar a_b$.
    Let $\{ t' \} = \semT{\bar \chi_a}(t_a,t_l)$.
    $t'$ then, by construction, contains an action $\unlock(m_g)$ that happens later than the $\unlock(m_g)$ in $t_l$.
    As $m$ is held at the end of $t_a$ and executing the sequence $\bar a_a$
    consisting of locking $m_g$, accessing $g$, and unlocking $m_g$ does not unlock $m$, at the end of $t'$ the ego thread
    still holds $m$.
    We distinguish two cases:
    \begin{itemize}
        \item Both $\bar \chi_a$ and $\bar \chi_b$ are executed by the same thread, with $m$ continuously held.
        \item The thread executing $\bar \chi_a$ has unlocked $m$ after this access,
            and this unlock precedes the last lock of $m$ by the ego thread in $t_b$.
    \end{itemize}
    In either case, $t'$ is a subtrace of $t_b$.
    Now consider executing the accesses in the reverse order.
    Then, $\semT{\bar \chi_b}(t_b,t_l) = \emptyset$, as $t_b$ contains $t'$
    and the $\unlock(m_g)$ in which $t_l$ ends is not the last unlock of $m_g$.
    Therefore, also $\semT{\bar \chi_a}(t_a,\semT{\bar \chi_b}(t_b,t_l)) = \emptyset$, and the left-hand side of the implication is false.
    %
    % The other case proceeds analogously.
    $\hfill\qed$
\end{proof}
\end{section}

\section{Predicate $||^?$ for Thread-ID Digests}\label{s:thread-ids}
A central instance of the digest framework are thread \emph{id}s~\cite{hammerandnail}.
A digest $\Digests$ that assigns to each thread an \emph{id} is called a \emph{thread id} digest.
As the digest may contain information other than the thread \emph{id} of the current thread,
we require that there is a function $\tid: \Digests \to \TIDs$ that, given a digest, yields its thread \emph{id}.
We demand that the thread \emph{id} of a given thread must remain constant during its execution.
Again, as in previous work~\cite{hammerandnail}, we assume the existence of two helper functions \mayrun\ and \unique:
\begin{itemize}
	\item $\unique : \TIDs \to \Bool$, yields, given a thread \emph{id}, a boolean indicating whether the thread represented by the \emph{id} is unique.
	\item $\mayrun : \Digests \to \Digests \to \Bool$, where $\mayrun\,a\,b$ must be true, in case the thread with digest $b$ may have already started
	when the ego thread has the digest $a$.
\end{itemize}

\begin{example}\label{ex:thread-ids}
	Consider the lightweight thread \emph{id} digest in \cref{f:threadingmainA}.
	We define the set of thread \emph{id}s $\TIDs$ as $\{\main, \othertid\}$ to allow differentiating between the $\main$ thread and other threads.
	The function $\tid$ is given by:
	\[
		\tid\;d = \begin{cases}
			\maintid & \text{if } d = \STmain \lor d = \MTmain \\
			\othertid & \text{otherwise}
		\end{cases}.
	\]
	It extracts the thread \emph{id} $\maintid$ out of the digests $\STmain$ and $\MTmain$.
	The function $\unique$ indicates that $\maintid$ is unique:
	\[
		\unique\;i = \begin{cases}
			\true & \text{if } i = \maintid \\
			\false & \text{otherwise}
		\end{cases}.
	\]
	The function $\mayrun$ is defined as follows:
	\[
		\mayrun\;a\;b = \begin{cases}
			\false & \text{if } a = \STmain \land b = \MT \\
			\true & \text{otherwise}
		\end{cases}.
	\]
	In case the ego thread has the digest $\STmain$, it can be excluded that a thread with digest $\MT$ has already been created.
	$\hfill\qed$
\end{example}
We now define the predicate $||^?$ for $\Digests$ generically in terms of these helper functions:
\begin{equation}\label{eq:general-thread-id-pred}
	a\;||^?\;b = \begin{cases}
		\false & \text{if } (\tid\;a = \tid\;b \land \unique\;(\tid\;a))\\
		& \qquad \lor \neg \mayrun\;a\;b\; \lor \neg \mayrun\;b\;a\\
		\top & \text{otherwise}
	\end{cases}
\end{equation}
This definition excludes races in case the digests $a$ and $b$ refer to the same thread \emph{id} and that the thread \emph{id} of $a$ is unique.
This is sound as a single thread cannot race with itself.
Additionally, it excludes races when either the thread with digest $b$ may not be started when the ego thread has the digest $a$, or vice versa.
\begin{example}
	Consider again the lightweight thread \emph{id} digest with the definitions in \cref{ex:thread-ids}.
	For this example, the definition in \cref{eq:general-thread-id-pred} coincides with the definition already given in \cref{ex:threadingextra}.
	$\hfill\qed$
\end{example}
The definition in \cref{eq:general-thread-id-pred} is also applicable to more intricate thread \emph{id} digests, such as
the ones proposed in earlier work~\cite[Section 6]{schwarz2023clustered}, where unique identities are also maintained for threads other than $\maintid$.

\fi

\end{document}